\documentclass[letterpaper, 11pt]{article}

\usepackage{tikz}
\usepackage{url}

\usepackage{graphicx}
\usepackage{latexsym}
\usepackage{fancyhdr}

\usepackage{amsmath,amsfonts,amssymb,bm, verbatim,mathtools}

\usepackage{xspace}

\newenvironment{proof}{\paragraph{Proof:}}{\hspace*{\fill}\(\Box\)}
\newtheorem{theorem}{Theorem}
\newtheorem{lemma}{Lemma}

\newcommand{\bibdir}{/Users/nitin/bibliography}

\newcommand{\GG}{\mathcal{G}}
\newcommand{\EE}{\mathcal{E}}
\newcommand{\PP}{\mathcal{P}}
\newcommand{\CC}{\mathcal{C}}

\newcommand{\inbound}{\mbox{\em inbound}}
\newcommand{\outbound}{\mbox{\em outbound}}

\newcommand{\scripte}{{\mathcal E}}

\newcommand{\Q}{Q}

\begin{document}

\centerline{\LARGE \bf Efficient Timestamps for Capturing Causality}

~

\begin{center}
Nitin H. Vaidya, University of Illinois at Urbana-Champaign\\
Sandeep S. Kulkarni, Michigan State University
\\~\\
\end{center}

~\\

\begin{center}
{\bf June 17, 2016}\footnote{\small A version of this report, excluding Section \ref{ss:sm}, was submitted for review to a conference
on May 11, 2016.}

~\\

\end{center}

\centerline{\bf Abstract}

Consider an asynchronous system consisting of processes that communicate via
message-passing.
The processes communicate over a potentially {\em incomplete} communication network consisting of reliable bidirectional communication channels.
Thus, not every pair of processes is necessarily able to communicate with each other directly.

The goal of the algorithms discussed in this paper is to assign timestamps
to the events at all the processes such that (a) distinct events are assigned distinct
timestamps, and (b) the happened-before relationship between the events
can be inferred from the timestamps. 
We consider three types of algorithms for assigning timestamps to events:
(i) Online algorithms that must (greedily) assign a timestamp to each event when
the event occurs. (ii) Offline algorithms that assign timestamps to event
after a finite execution is complete. (iii) Inline algorithms that assign a timestamp to each event when it occurs, but may modify some elements of a timestamp again at a later time.

For specific classes of graphs, particularly {\em star} graphs and graphs with connectivity $\geq 1$,
the paper presents bounds on the length of vector timestamps assigned by an {\em online}
algorithm. The paper then presents an {\em inline} algorithm, which typically 
assigns substantially smaller timestamps than
the optimal-length {\em online} vector timestamps. In particular, the inline algorithm assigns timestamp in the form of a tuple containing $2c+2$ integer elements, where $c$ is the size of the vertex cover for the underlying communication graph.

~\\

%
%

\newpage

\section{Introduction}
\label{sec:intro}

Consider an asynchronous system consisting of $n$ processes that communicate via
message-passing.
The processes communicate over a potentially {\em incomplete} network of  reliable bidirectional communication channels.
%
The goal of the algorithms discussed in this paper is to assign timestamps
to the events at all the processes such that (a) distinct events are assigned distinct
timestamps, and (b) the happened-before \cite{lamport78} relationship between the events
can be inferred from the timestamps. 

We will consider three types of algorithms for assigning timestamps to events. To allow us to
compare their behavior, let us introduce a {\em query} abstraction for timestamps. 
For event $e$, we use $e_r$ to denote the abstract real time (which is not available to the processes themselves) when event $e$ occurred. 
The timestamp of event $e$ may be queried at any real time $t$, $t \geq e_r$. 
Depending
on the timestamp algorithm in use, the query may or may not return immediately.
Denote by $\Q^t(e)$ the timestamp that would be returned if a query {\em were to be issued}
at real time $t$ for the timestamp of event $e$. Note that $\Q^t(e)$ is defined even if no query
is actually issued at time $t$. 
Also note that $\Q^t(e)$ is only defined if $t \geq e_r$.
The delay in computing $\Q^t(e)$ depends on the algorithm for assigning timestamps, as seen
below.
Now let us introduce three types of timestamp algorithms:

\begin{itemize}
\item Online algorithms: An online algorithm must (greedily) assign a distinct timestamp to each event
{\em when the event occurs}.
Suppose that $\tau(e)$ is the timestamp assigned to event $e$ by an online algorithm.
The assigned timestamps must be such that,
for any two events $e$ and $f$, $e\rightarrow f$ iff $\tau(e)<\tau(f)$, where 
$<$ is a suitably defined partial order on the timestamps, and $\rightarrow$ is the happened-before
relation \cite{lamport78}. The {\em vector timestamp} algorithm \cite{DBLP:journals/computer/Fidge91,Mattern88virtualtime} is an example of an online algorithm.
For an online algorithm, for any event $e$, $\Q^t(e)=\tau(e)$ for $t\geq e_r$; thus, a query issued at time $t\geq e_r$ can immediately return $\Q^t(e)$.

\item Offline algorithms: An offline algorithm takes an entire (finite) execution
as its input, and assigns a distinct timestamp to each event in the execution.
Similar to online algorithms, the timestamps $\tau(.)$ must be such that,
for any two events $e$ and $f$, $e\rightarrow f$ iff $\tau(e)<\tau(f)$, where 
$<$ is a suitably defined partial order on the timestamps.
There is significant past work on such offline computation of timestamps \cite{DBLP:journals/ipl/Charron-Bost91}.
%
For an offline algorithm, query for the timestamp of any event
will {\em not} return until the entire (finite) execution is complete, and the
offline algorithm has subsequently computed the timestamps for the events.

\item Inline algorithms: Timestamp assigned to an event by an inline algorithm may change as the execution proceeds.
Thus, for an event $e$ it is possible
that $\Q^{t1}(e)\neq \Q^{t2}(e)$ for $t2>t1\geq e_r$.
However, timestamps for distinct events must be always distinct.
That is, for distinct events $e$ and $f$,
$\Q^{t}(e)\neq \Q^{v}(f)$ for any $t \geq max(e_r, f_r)$
We refer to timestamps assigned by inline algorithms as {\em inline timestamps}.
A suitable partial order $<$ is defined on the inline timestamps.
The inline timestamps must satisfy the following property for any two events
$e$ and $f$  and for any $t \geq max(e_r, f_r)$,

\hspace*{1.5in} $\Q^t(e)<\Q^t(f)$ if and only if $e\rightarrow f$.

Thus, the timestamps returned to queries at time $t$
must capture happened-before relation between events that have occurred by that
time; however, $\Q^t(e)$ may not suffice to infer happened-before relation with
some other event $g$ where $g_r > t$.
As an example, suppose that $e\rightarrow g$. Then, it is possible that $\Q^{e_r}(e)\not<\Q^{g_r}(g)$; however, as noted above, it must be true that $\Q^{g_r}(e)<\Q^{g_r}(g)$.

\end{itemize}

The {\em inline} algorithm presented in Section \ref{sec:inline} has close similarities to mechanisms
introduced previously for shared memory \cite{DBLP:journals/tocs/LadinLSG92,zawirski2015swiftcloud}
and message-passing \cite{DBLP:conf/icdcs/GargS02,DBLP:journals/dc/GargSM07}.
We will elaborate on the similarities (and differences) later in Section \ref{sec:related}.
Despite the past work, it appears that
the ideas presented here have some novelty,
as elaborated in Section \ref{sec:related}.

 
\section{System Model and Notation}
\label{sec:notation}

We consider an asynchronous system. The $n$
processes in the system are named $p_i$, $0\leq i<n$.
Processes communicate via reliable bidirectional message-passing channels.
The communication graph for the system includes only undirected edges,
and is denoted by $\GG(\PP,\EE)$.
$\PP = \{p_0,p_1,\cdots,p_{n-1}\}$ denotes the set of vertices, where
vertex $p_i$ represents process $p_i$.
$\EE$ is the set of undirected edges, where the undirected edge between $p_i$ and $p_j$, $p_i \neq p_j$, represents a
bidirectional link.

The events are of three types: send events, receive events,
and computation events.
We consider only {\em unicasts}, thus, each send event results in a message sent to exactly
one process.
However, the proposed algorithm can be easily adapted when multiple messages may be sent at a single send event.

For an event $e$,
{\em proc(e)} denotes the process at which event $e$ takes place.
$\rightarrow$ denotes the happened-before relation between events \cite{lamport78}.
For events $e$ and $f$, when $e\rightarrow f$, we say that ``$f$ happened-after $e$''.
If $e$ is a send event, then $recv(e)$ is the receive event at the recipient process
for the message sent at event $e$.
Recall that for event $e$, $e_r$ denotes the real time at which $e$ occurs.
Different events at the {\em same} process occur at different real times.
That is, if $e\neq f$ and $proc(e)=proc(f)$,  then $e_r\neq f_r$.

For an event $e$, {\em index(e)} denotes the index of event $e$ at process {\em proc(e)}.
For convenience, define $index(\perp)=\infty$.
In Figure \ref{fig:cover}(a), for event $g$, {\em proc(g)} = $p_3$
and {\em index(g)} = 2, because $g$ is the second event at $p_3$.
%
%
For an event $e$ at a process $p_i$, and a process $p_j\neq p_i$, we define events $\outbound(p_j,e)$ and $\inbound(p_j,e)$
as follows:
\begin{itemize}
\item
$\outbound(p_j,e)$ at time $t\geq e_r$ denotes the event at $p_i$ where $p_i$ sends the first message
to $p_j$ at or after event $e$. If $p_i$ has not sent such a message by time $t$, then
$\outbound(p_j,e)=\perp$ at time $t$.
In particular, if $p_i$ sends a message to $p_j$ at event $e$, then $\outbound(p_j,e)=e$;
otherwise, after event $e$, if process $p_i$ sends the first message to $p_j$ at some event $f$ such that $f_r\leq t$,
then $\outbound(p_j,e)=f$ at time $t$.



\item
$\inbound(p_j,e)$ at any time $t\geq e_r$ is defined as follows:
(i) If $\outbound(p_j,e)=\perp$ at time $t$, then $\inbound(p_j,e)=\perp$.
(ii) Else $\inbound(p_j,e)=recv(\outbound(p_j,e))$. 
It is possible that, even when $\outbound(p_j,e)\neq\perp$ at time $t$, 
the receive event $recv(\outbound(p_j,e))$ may not yet be known
-- thus, although
$\inbound(p_j,e)$ may be
well-defined at time $t$, its value (i.e., $recv(\outbound(p_j,e))$) may not be known until later. As we will see later, this affects the design of the {\em inline} algorithm in Section \ref{sec:inline}.
\end{itemize}





Let ${\bf 0}$ denote a vector with all elements being 0; size of the vector is determined
by the context.
Similarly, let $\boldsymbol{\infty}$ denote a vector with all elements being $\infty$.
$V[j]$ denotes element of vector $V$ at index $j$.
Unless stated otherwise, for a vector of length $m$, we index its elements
as 0 through $m-1$.
For vectors $U$ and $V$ of equal size, $\max(U,V)$ is a vector obtained by taking
their element-wise maximum. That is, the $j$-th element
of vector $\max(U,V)$ equals $\max(U[j],V[j])$.

\section{Vector Timestamps with Online Algorithms}
\label{sec:vector}

The proposed {\em inline} algorithm in Section \ref{sec:inline} assigns timestamps
whose size depends on the vertex cover for the communication graph.
A vertex cover of $\GG(\PP,\EE)$ is a subset $\CC$ of $\PP$ such that each edge in $\EE$ is incident on at least one vertex in $\CC$.
In particular, consider
a {\em star} graph in which each process $p_i$, $i\neq 0$, has a link only with process $p_0$; there are no other links in a star graph. $p_0$ is the {\em central} process of the star graph, other processes being {\em radial} processes.
The {\em star} graph has a vertex cover $\{p_0\}$ of size 1, and thus, the proposed {\em inline} algorithm
assigns the smallest timestamps for star graphs. For comparison, we now present some
bounds on timestamps assigned by {\em online} algorithms for some special classes of
graphs, including {\em star} graphs.
Let us define vector timestamps formally \cite{DBLP:journals/ipl/Charron-Bost91}.

\paragraph{Vector timestamps:}
Suppose that a given online algorithm assigns to each event $e$ a timestamp $\tau(e)$
consisting of a vector of a certain fixed size. These timestamps are said to be {\em vector timestamps}
provided that $\tau(e)<\tau(f)$ iff and only if $e\rightarrow f$, where the partial
order $<$ on timestamp vectors (such as $\tau(e)$ and $\tau(f)$) is defined as follows:
{\em
For vectors $U$ and $V$, $U<V$ iff (a) $\forall j$ $U[j]\leq V[j]$, and (b) $\exists i$ such that $U[i]<V[i]$.
}

Vector timestamps are well-studied, and it has been shown that, in general,
the vector length must be at least $n$ in the worst case even if
the timestamps are assigned by an {\em offline} algorithm \cite{DBLP:journals/ipl/Charron-Bost91}.

For {\em online} algorithms, and special classes of graphs, we show the following
bounds on the length
of the vector timestamps necessary to capture causality (i.e., $\tau(e)<\tau(f)$ iff $e\rightarrow f$).
It appears that these bounds have not been obtained previously.
%
%
\begin{itemize}
\item {\em Star graphs:}

\begin{itemize}
\item {\em Real-valued vector elements}:
For $n\geq 3$,
when the vector elements may take any finite real-value,  $n-1$ is the tight bound
for vector timestamp length for the {\em star} communication graph when using an {\em online} algorithm.
Lemma \ref{l_bound} presented at the end of this section proves the lower bound of $n-1$, and 
Appendix \ref{app:upper:star} presents an online algorithm, which constructively
proves that $n-1$ is also an upper bound for $n\geq 3$.

\vspace*{4pt}

For $n=2$, vector length of 2 can be shown to be necessary and sufficient.

\vspace*{2pt}

\item {\em Integer-valued vector elements}:
When the vector elements are constrained to take integer values, $n$ is the tight bound
for vector timestamp length for the {\em star} communication graph when using an {\em online} algorithm.
Lemma \ref{l_bound_integer} in Appendix \ref{app:lower:star} proves the lower bound of $n$.
Upper bound of $n$ is achieved by the standard vector clock algorithm \cite{DBLP:journals/computer/Fidge91,Mattern88virtualtime}.

\end{itemize}

\item {\em Graphs with vertex connectivity = $\kappa$:}
\begin{itemize}
\item Vertex connectivity {\em $\kappa \geq 2$}:
For any communication graph with vector connectivity $\kappa\geq 2$, Lemma \ref{l_bound_arbitrary} in Appendix \ref{sec:bounds:arbitrary}
proves that an online algorithm must use a vector timestamp of length $n$
in the worst case.
Upper bound of $n$ is achieved by the standard vector clock algorithm \cite{DBLP:journals/computer/Fidge91,Mattern88virtualtime}.

\vspace*{2pt}

\item Vertex connectivity {\em $\kappa = 1$}:
For any given communication graph with vertex connectivity of $\kappa=1$, define $X$ to be the set
of processes such that no process in set $X$ by itself forms a vertex cut of size 1.
Then, as shown in Lemma \ref{l_conn1} in Appendix \ref{sec:bounds:arbitrary}, $|X|$ is a lower bound on the vector size used by an online algorithm.
Note that for the star graph, $|X|=n-1$.

\end{itemize}
\end{itemize}

%
%
%
%

While the above results for star graph show that it is not possible to assign
small timestamps using {\em online} algorithm, we presently do not know if
a similar claim is true for  {\em offline} algorithms in {\em star} graphs.
Appendix \ref{a:offline} presents a preliminary result for $n=4$ that suggests 
that further investigation is necessary to resolve the question.

\begin{lemma}
\label{l_bound}
Suppose that an {\em online} algorithm for the star graph assigns distinct {\bf real-valued} vector timestamps
to distinct events such that, for any two events $e$ and $f$, $e\rightarrow f$ if and only if $\tau(e)<\tau(f)$.
Then the vector length must be at least $n-1$.
\end{lemma} 
\begin{proof}
The proof is trivial for $n\leq 2$. Now assume that $n\geq 3$.
The proof is by contradiction. Suppose that a give online algorithm assigns vector timestamps of length $s\leq n-2$.

Let $e_q^j$ denote the $q$-th event at process $p_j$.
Consider an execution that includes a send event
$e_1^i$ at radial process $p_i$, $1\leq i\leq n-1$,
where the radial process $p_i$ sends a message to the central process $p_0$.
These $n-1$ send events are concurrent with each other. 
At process $p_0$, there are $n-1$ {\em receive} events corresponding to the above
{\em send} events at the other processes.
The execution contains no other events.

$\tau(e_q^j)$ denotes the vector timestamp of length $s$ assigned to event $e_q^j$ by the online algorithm.
Create a set $S$ of processes as follows:
for each $l$, $0\leq l< s$,
add to $S$ any one radial process $p_j$ such that $\tau(e_1^j)[l] = \max_{1\leq i<n}\, \tau(e_1^i)[l]$.
Note that $\tau(e_1^i)[l]$ is the $l$-th element of vector $\tau(e_1^i)$.
Clearly, $|S|\leq s\leq n-2$.
Consider a radial process $p_k\not\in S$ (note that $p_k\neq p_0$). Such a process
$p_k$ must exist since $|S|\leq n-2$, and there are $n-1$ radial processes.

Suppose that the message sent by process $p_k$ at event $e_1^k$ reaches process $p_0$ after all
the other messages, including messages from all the processes in $S$, reach process $p_0$. That is,
$e_{n-1}^0$ is the receive event for the message sent by process $p_k$. By the time
event $e_{n-1}^0$ occurs, (online) timestamps must have been assigned to all the other events in this
execution.
This scenario is possible because the message delays can be arbitrary, and
 an online algorithm assigns timestamps to the events when they occur.

Now consider event $e_{n-2}^0$. By event $e_{n-2}^i$,
except for the message sent by process $p_k$, all the other messages, including messages
sent by all the processes in $S$, are received by process $p_0$.

Define vector $E$ such that $E[l] = \max_{1\leq i<n}\, \tau(e_1^i)[l]$, $0\leq l< s$.
By definition of $S$, we also have that
$E[l] = \max_{p_i\in S}\, \tau(e_1^i)[l]$, $0\leq l< s$.
The above assumption about the order of message delivery implies that $E \leq \tau(e_{n-2}^0).$
Also, since $p_k\not\in S$, we have that $\tau(e_1^k)\leq E.$
The above two inequalities together imply that $\tau(e_1^k)\leq \tau(e_{n-2}^0)$.

Since $e_1^k\neq e_{n-2}^0$, their timestamps must be distinct too. Therefore,
$\tau(e_1^k)< \tau(e_{n-2}^0)$, which, in turn, implies that
$e_1^k\rightarrow e_{n-2}^0$.
However, $e_1^k$ and $e_{n-2}^0$ are concurrent, leading to a contradiction.
\end{proof}

\section{Inline Algorithm}
\label{sec:inline}

The structure of the inline timestamps presented here has close similarities to comparable objects introduced
in past work, in the context of message-passing \cite{DBLP:conf/icdcs/GargS02,DBLP:journals/dc/GargSM07}
 and causal memory systems \cite{DBLP:journals/tocs/LadinLSG92,zawirski2015swiftcloud}. We discuss
the related work in Section \ref{sec:related}, and also describe the extra flexibility offered by
our approach.

The proposed inline algorithm makes use of a vertex cover for the given communication graph.
Let $\CC$ be the chosen vertex cover. It is assumed that each process knows the cover set $\CC$.
Define $c=|\CC|$.  Without loss of generality, suppose that the processes
are named such that $\CC=\{p_0,p_1,\cdots,p_{c-1}\}$.

The algorithm assigns a timestamp $\tau(e)$ to each event $e$.
The timestamp for an event at each process $\CC$ consists of just a vector of size $c$.
On the other hand, the timestamp for an event at a process outside $\CC$ includes other
components as well. We refer to the vector component in a timestamp $\tau(e)$ as $\tau(e).vect$.
The other components of the timestamp assigned to an event outside $\CC$ are $id$,
$index$ and $next$ (elaborated below).

The algorithm assigns an initial timestamp to each event $e$ when the event occurs.
The {\em vect} field of a timestamp is not changed subsequently.
Similarly, the {\em index} field, present only in timestamps of events outside $\CC$,  is also not changed subsequently.
The $next$ field of the timestamp, assigned only to an event outside $\CC$, however, may be updated as the execution progresses beyond the event
(as elaborated below).
Since the timestamps for events in $\CC$ only include the $vect$ field, it follows that the
once a timestamp is assigned to an event in $\CC$, it is never modified.


\paragraph{Intuition behind inline timestamps:} 
The inline algorithm exploits the fact that
at least one endpoint of each communication channel must be at a process in $\CC$.
In particular, for an event $e$ that occurs at a process outside $\CC$, the algorithm identifies the {\em most recent}
event in $\CC$, say $f$, such that $f\rightarrow e$. Similarly, for an event $e$ that occurs outside $\CC$,
the algorithm identifies the {\em earliest} event at each $p_j\in\CC$, say event $f_j$, that happened-after $e$
and is influenced {\bf directly} by the process where $e$ occurs. Here ``influence directly'' means that
the process $proc(e)$ sends a message to $p_j$. Indices of these events are used to form the {\em inline}
timestamp of event $e$. Since $proc(e)$ may ``directly influence'' different processes at different times, the corresponding
components of the timestamp are updated accordingly when necessary. 

For events at processes in $\CC$, the inline algorithm uses the standard vector clock algorithm \cite{DBLP:journals/computer/Fidge91,Mattern88virtualtime}, with the
vector elements restricted to processes in $\CC$.
In particular,
for an event $e$ at $p_i\in\CC$, $\tau(e).vect$ is a vector of length $c$, and
with the following properties:
\begin{itemize}
\item If $e$ is the $k$-th event at $p_i$, then $\tau(e).vect[i]=k$.
\item For $p_j\in\CC$ where $p_j\neq p_i$, $\tau(e).vect[j]$ is the number of events at $p_j$ that happened-before $e$.
\end{itemize}

\subsection{Inline Algorithm Pseudo-Code}

Each process $p_i$ maintains a local vector clock $clock_i$ of size $c$. Initially,
$clock_i := {\bf 0}$.
Consider a new event $e$ at process $p_i$.
We now describe how the various fields of the timestamp are computed:

\begin{itemize}
\item If $p_i\not\in\CC$ then $\tau(e).id:=p_i$ and $\tau(e).index:=index(e)$.
\item {$vect$ field:}
\begin{itemize}

\item If $p_i\in \CC$, then $clock_i[i] := clock_i[i] + 1$.

\item If $e$ is a send event, then piggyback
the following on the message sent at event $e$: (i) vector $clock_i$,
and (ii) if $p_i\not\in\CC$ then $\tau(e).index$.

\item If $e$ is a receive event, then let
$v$ be the vector piggybacked with the received message, and update $clock_i := \max(clock_i,v).$

\item
$\tau(e).vect := clock_i$.
\end{itemize}

\item $next$ field: If $p_i\in\CC$, computation of $next$ is not performed.

The steps performed when $p_i\not\in\CC$ depend on the type of the event, as follows:
\begin{enumerate}

%
%
%
%
%
%

\item $\tau(e).next:= \boldsymbol{\infty}$.

\item If $e$ is a send event for message\footnote{Because $p_i\not\in\CC$
and $\CC$ is a vertex cover,
any message from $p_i$ must be sent to a process in $\CC$.}
memor $m$ destined for some process $p_j\neq p_i$
then define an event set $N_e$ as follows: 
\[ N_e ~=~\{e\} \cup \{ f~|~\mbox{$proc(f)=p_i$ and $f\rightarrow e$ and~}\tau(f).next[j]=\infty\}\]
\item 
When $index(\inbound(p_j,e))$ becomes known to $p_i$,
for each $g\in N_e$, $$\tau(g).next[j] := index(\inbound(p_j,e))$$

The discussion of how $p_i$ learns $index(\inbound(p_j,e))$ is included with the
discussion of the query procedure in Section \ref{sec:query}. 

\end{enumerate}
\end{itemize}
Observe that the algorithm essentially assigns
vector timestamps to events in $\CC$, with vector elements restricted to the
processes in $\CC$. The $next$ field for events outside $\CC$ may change
over time, as per steps 2 and 3 above. 
\subsection{Response to a Query for Timestamps}
\label{sec:query}

Consider any event $e$ that occurs at time $e_r$ at some process $p_i$. If by
some time $v\geq e_r$, the event $\outbound(p_j, e)$ has occurred already, but $\tau(e).next[j]=\infty$,
then the $next[j]$ field of timestamp $Q^v(e)$ of $e$ cannot yet be determined  (refer to Step 3 of the algorithm above).
Hence, the query for $Q^v(e)$ is delayed until this information becomes available to $p_i$.
To allow $p_i$ to learn the index of the receive event for the message it sent to $p_j$ at event $\outbound(p_j, e)$,
process $p_j$ can send a control message to $p_i$ carrying
the index of its receive event, as well as the index of the corresponding send
event at $p_i$ (the index of the send event is piggybacked on the application message,
as specified in the pseudo-code above). 
Dashed arrows in Figure \ref{fig:cover} illustrate
such control messages. In particular, the last control message in Figure \ref{fig:cover}(b) carries index 5 of the receive event at $p_1$ and index 4 of the
corresponding send event at $p_3$. Section \ref{sec:example} elaborates on the example in Figure \ref{fig:cover} 

The overhead of the above control messages can potentially be mitigated by judiciously piggybacking control information on application messages. Alternatively,
the control information can be ``pulled'' only when needed. In particular,
when a query for timestamp of some event $e$ is performed at $p_i$ at time $v$,
$p_i$ can send a control message to the processes in $\CC$ to learn any
event index information that may be necessary to return $Q^v(e)$. 

To summarize, the response to a query for timestamps
of event $e$ at process $p_i\not\in\CC$ at time $v\geq e_r$ is handled as follows:

\begin{itemize}
\item While ($\exists p_j\in\CC$ such that $\outbound(p_j,e)\neq \perp$, and $\tau(e).next[j]=\infty$) wait.
\item Return $\tau(e)$ as $Q^v(e)$.
\end{itemize}

\subsection{Example of Inline Timestamps and Query Procedure}
\label{sec:example}

\begin{figure}[ht]
\centering
\centerline{\includegraphics[width = 2.2in]{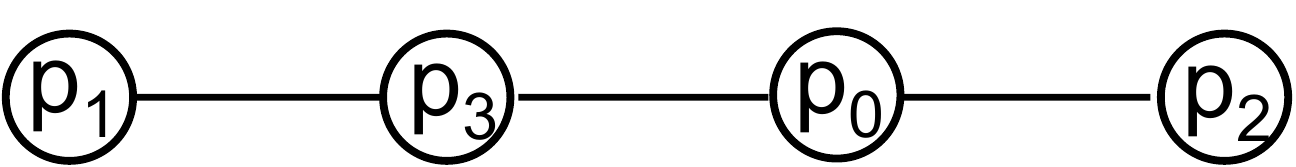}}
\caption{\em An example communication graph}
\label{fig:topology}
\end{figure}

Consider the communication network in Figure \ref{fig:topology}. For this network, let us choose $\CC=\{p_0,p_1\}$.
Thus, the timestamps for events at processes in $\CC$ (i.e., $p_0$ and $p_1$)
consist of a vector of length 2. 
Figure \ref{fig:cover}(a) shows all the events that have
taken place in a certain execution by time $t$. In this execution,
the initial timestamp $Q^{e_r}(e)$
for event $e$ at $p_0$ is (3,1) because it is the third event at process $p_0$,
and only one event at $p_1$ happened-before event $e$ (this dependence arises
due to messages exchanged by $p_0$ and $p_1$ with process $p_3\not\in\CC$).
The timestamp for an event in $\CC$ does not change after the initial assignment. Thus, $Q^{e_r}(e)=Q^t(e)$ for any $t > e_r$. 
The solid arrows in the figure depict application messages, whereas
the dashed arrow depicts a control message, to be explained later.

\begin{figure}[ht]
\centering
{\includegraphics[height = 2.1in]{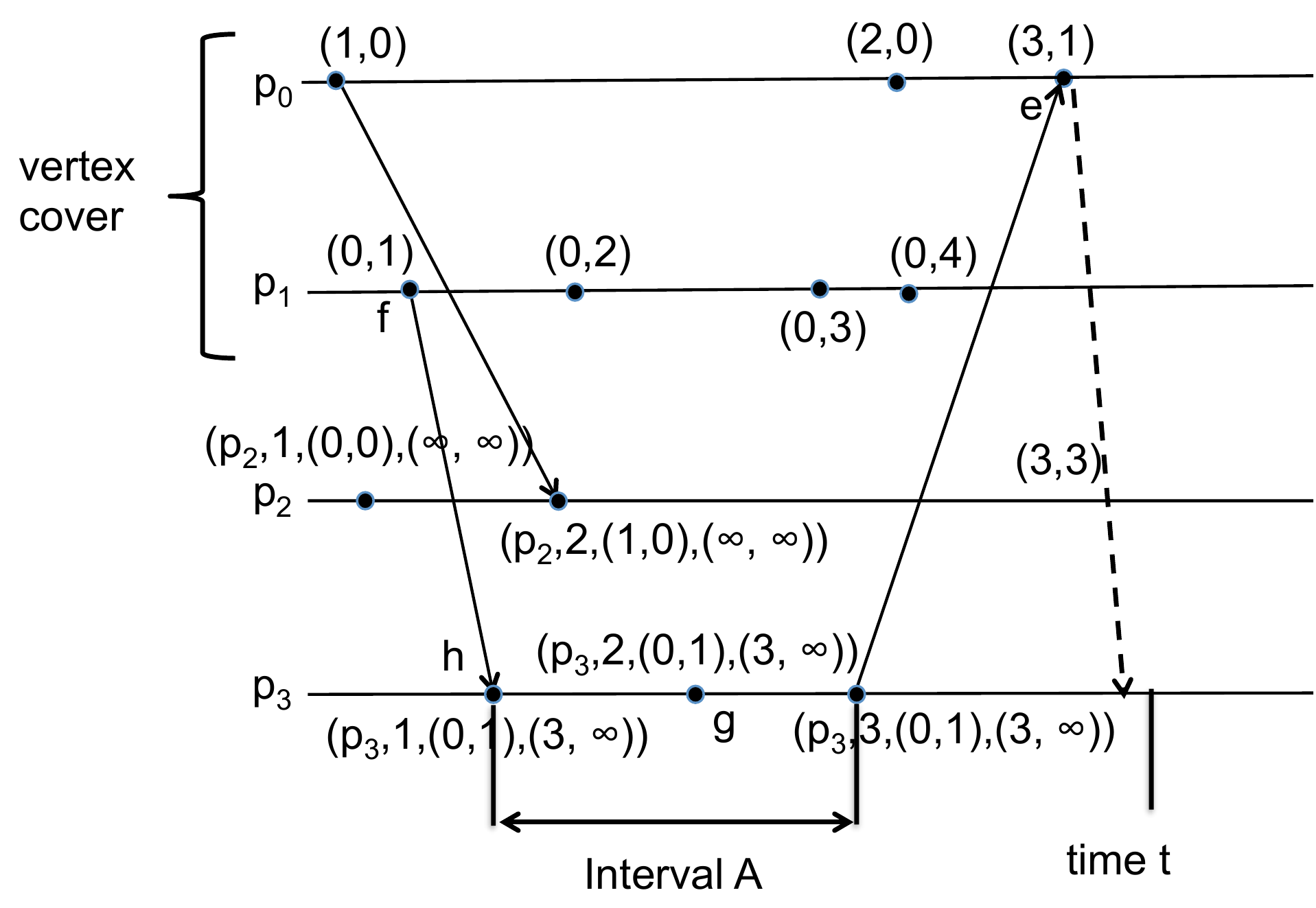}}
{\includegraphics[height = 2.1in]{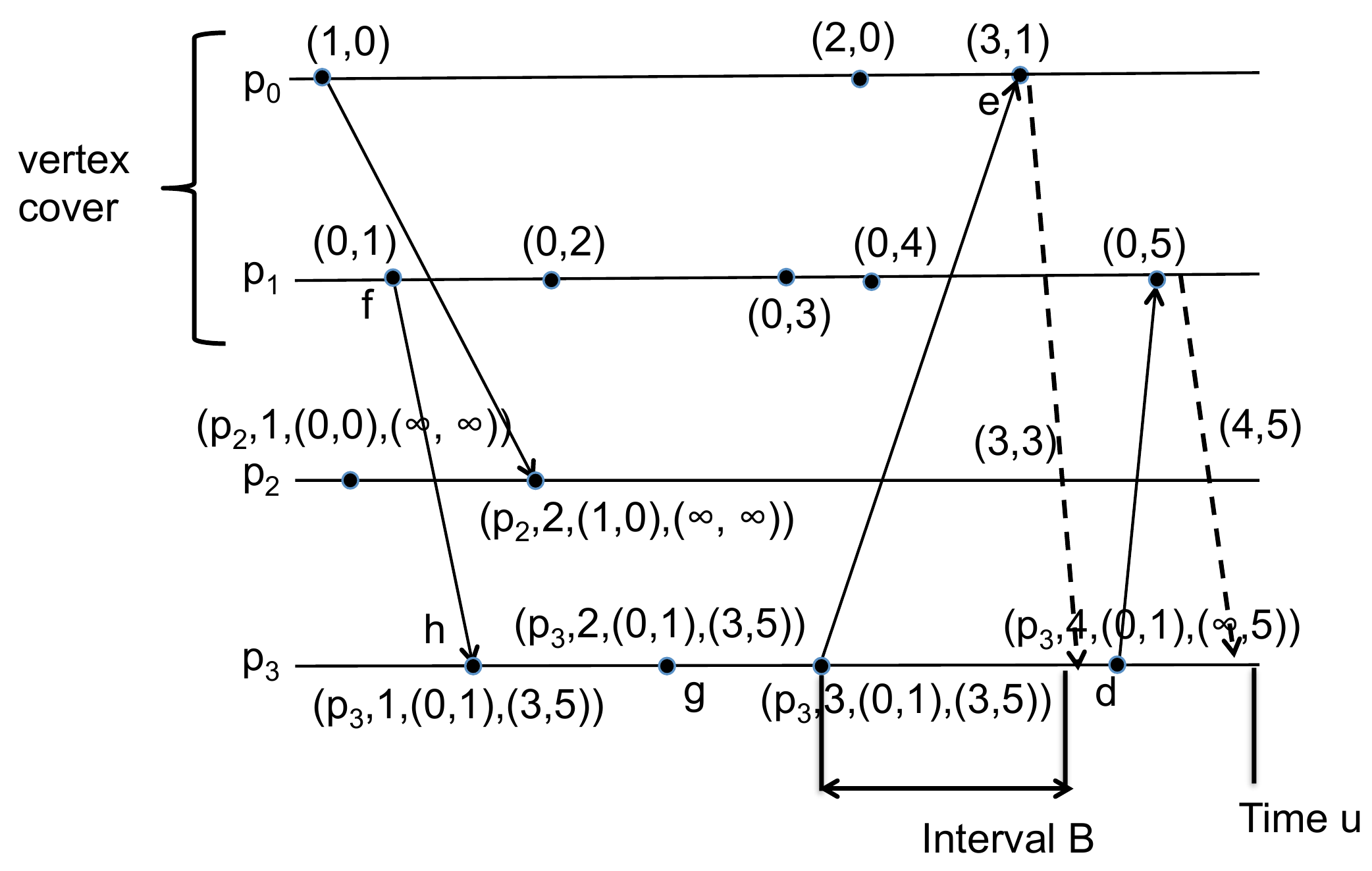}}

(a) \hspace*{2.8in}(b)
\caption{\em (a) An execution showing all the events that have taken place by time $t$. For each event $x$ in the figure, timestamp $Q^t(x)$ returned for a query issued at $t$ is also shown.
A solid arrow depicts an application messages, whereas
dashed arrow depicts a control message.
(b) Extended execution 
showing events that have taken place by time $u$, and timestamp $Q^u(x)$ for each event $x$.  }
\label{fig:cover}
\end{figure}

In Figure \ref{fig:cover}(a), the timestamp $Q^{h_r}(h)$
of event $h$ when it occurs at $p_3$ at time $h_r$
is $(p_3,1,(0,1),(\infty,\infty))$ -- this timestamp is not depicted in the figure.
However, at time $t$, as depicted in Figure \ref{fig:cover}(a),
$Q^t(h)= (p_3,1,(0,1),(3,\infty))$. The $index$ in the timestamp is 1 because $h$ is the
first event at $p_3$. $vect=(0,1)$ in the timestamp because
no event at $p_0$ and 1 event at $p_1$ happened-before event $h$.
Observe that $vect$ field of $Q^{h_r}(h)$ and $Q^t(h)$ is identical.
In fact, except the $next$ field, the other fields of the timestamps
do not change after their initial assignment.
Both the elements of $next$ in $Q^{h_r}(h)=(p_3,1,(0,1),(\infty,\infty))$ are $\infty$, because
$h$ is not a send event. Subsequently, if and when $p_3$ sends messages to processes in $\CC$,
corresponding elements of $next$ are updated. For instance,
the $next[0]$ element of $Q^t(h)$ is $3$ because 3 is the index of the event $e$ at $p_0$
at which $p_0$ receives a message from $p_3$ that was sent at time $\geq h_r$ and $\leq t$.
$next[1]$ in timestamp $Q^t(h)$ is $\infty$ because process $p_3$ does not
send a message to process $p_1$ at any time between $h_r$ and $t$.
Observe that the timestamp $Q^t(g)$ for event $g$ at $p_3$ differs from $Q^t(h)$
only in its $index$: the $vect$ and $next$ are identical for events $h$ and $g$.
The dashed arrow in Figure \ref{fig:cover}(a) depicts a {\em control message}
that allows process $p_3$ to learn
the index of the event at $p_0$ where $p_0$ received a message from $p_3$.
Process $p_3$ can determine, on receipt of the control message in Figure \ref{fig:cover}(a),
that $Q^t(h).next[0]=3$.

For any $j$, once $next[j]$ is assigned a finite value, the field $next[j]$ is not modified
again. For instance, in the above example, because $Q^t(h).next[0]=3$,
it follows that $Q^u(h).next[0]=3$ for any $u>t$ as well. 
However, since $Q^t(h).next[1]$ is $\infty$,
if at a later time, process $p_3$ were to send a message to $p_1$,
$next[1]$ is updated appropriately.
For instance, Figure \ref{fig:cover}(b) shows an extended version
of the execution in Figure \ref{fig:cover}(a)  that shows all the events
that occur by some time $u>t$.
Also, Figure \ref{fig:cover}(b) shows the timestamp $Q^u(x)$ corresponding to query at time $u$ for each event $x$ shown in the figure.
Now observe that $next[1]$ for event $h$ and $g$ are both changed to 5,
because the message sent by process $p_3$ to $p_1$ is received at the $5^{th}$ event at $p_1$.
For event $d$ at $p_3$, while $next[1]$ equals 5 (similar to $next[1]$ for event $g$),
$next[0]$ for event $d$ is presently $\infty$, since process $p_3$ is yet to send a message
to $p_0$ (i.e., at or after event $d$).

\paragraph{Delay in responding to some queries:}
In Figure \ref{fig:cover}(a), at any time $v$ during interval A, 
$Q^v(h)$ will be returned as $(p_3,1,(0,1)(\infty,\infty))$ because
process $p_3$ is yet to send any message after event $h$.
In Figure \ref{fig:cover}(a), the dotted arrow is a control message that carries
the index of event $e$ -- on receipt of this message, $p_3$ learns
$\inbound(p_0,h)$.
In Figure \ref{fig:cover}(b), query issued anytime during interval B will have 
to wait until $p_3$ learns $\inbound(p_0,h)$.
On the other hand, query at time $u$ in Figure \ref{fig:cover}(b) will return $Q^u(h)=(p_3,1,(0,1),(3,5))$.

\subsection{Inferring Happened-Before ($\rightarrow$) from the Inline Timestamps}
\label{sec:infer}

Recall that timestamps for events at processes in $\CC$ do not include
an $id$ field, whereas timestamps for events at 
processes in $\CC$ do include an $id$ field.
In the following, we
use the convention that, if $\tau_1$ is the timestamp for an event at
a process in $\CC$, then $\tau_1.id=\perp$.
On the other hand,
for timestamps of events at processes outside $\CC$, $id\neq\perp$.

For the inline timestamps defined above, we define the $<$ relation as follows.
Consider two inline timestamps $\tau_1$ and $\tau_2$.
$\tau_1 < \tau_2$ if and only if one of the following is true:
\begin{list}{}{}
\item
\hspace*{0.25in} (i)~ $\tau_1.id =\tau_2.id \neq \perp$ and $\tau_1.index < \tau_2.index$, or

\hspace*{0.25in} (ii) $\tau_1.id=\tau_2.id=\perp$ and $\tau_1.vect < \tau_2.vect$, or

\hspace*{0.25in} (iii) $\tau_1.id=\perp$, $\tau_2.id\neq\perp$, and $\tau_1.vect \leq \tau_2.vect$, or

\hspace*{0.25in} (iv) $\tau_1.id\neq\perp$, $\tau_1.id\neq\tau_2.id$, and 
$\exists i,~0\leq i<c$, such that $(\tau_1.next[i] \leq \tau_2.vect[i])$.
\end{list}

The four cases above cover all possibilities. In particular,
in case (i), the two events are at the same process outside $\CC$.
In case (ii), the two event are at processes (not necessarily identical) in $\CC$.
In case (iii), $\tau_1$ is timestamp of an event at a process in $\CC$,
whereas $\tau_2$ corresponds to an event outside $\CC$.
Finally, in case (iv), timestamp $\tau_1$ corresponds to an event at a process outside $\CC$,
whereas the event corresponding to $\tau_2$ may be at any other process (in or outside $\CC$).

With the above definition $<$,
the theorem below states that the inline algorithm satisfies the requirement
that the timestamps be useful in inferring causality.
\begin{theorem}
\label{thm1}
For any two events $e$ and $f$, and for $t \geq max(e_r, f_r)$,
$Q^t(e)$ and $Q^t(f)$ are the timestamps returned by the query procedure when
using the proposed inline algorithm. The following condition holds:\\
\hspace*{1in}$e\rightarrow f$~~ if and only if ~~ $Q^t(e)<Q^t(f)$,\\
where partial order $<$ for inline timestamps is as defined above.
\end{theorem}
Appendix \ref{sec:proof} presents the proof of this theorem.
Appendix \ref{sec:app} discusses some implementation issues related to the 
inline algorithm.

\section{Related Work}
\label{sec:related}


%

The concept of vector clock or vector timestamp 
was introduced by Mattern \cite{Mattern88virtualtime} and
Fidge \cite{DBLP:journals/computer/Fidge91}.
Charron-Bost \cite{DBLP:journals/ipl/Charron-Bost91} showed
that there exist communication patterns
that require vector timestamp length equal to the number of processes.
Schwarz and Mattern \cite{DBLP:journals/dc/SchwarzM94} provided a relationship between
the size of the vector timestamps and the dimension of the partial order specified by 
happened-before.
Garg et al. \cite{DBLP:conf/podc/GargS01} also demonstrated analogous
bounds on the size of vector timestamps using the notion of event {\em chains}.
Singhal and Kshemkalyani \cite{DBLP:journals/ipl/SinghalK92} proposed a strategy for reducing the
communication overhead of maintaining vector timestamps.
Shen et al. \cite{DBLP:conf/ispdc/ShenKK13} encode of a vector clock of length $n$ using
a single integer that has powers of $n$ distinct prime numbers as factors. 
Torres-Rojas and Ahmad propose constant size logical clocks
that trade-off clock size with the accuracy with which happened-before relation
is captured \cite{Torres-Rojas:1999:PCC:1035766.1035768}.
Meldal et al. \cite{DBLP:conf/podc/MeldalSV91} propose
a scheme that helps determine causality between two messages {\em sent to the same process}.
They observe that, because their timestamps do not need to capture the happened-before relation between {\em all events}, their timestamps can be smaller. Some of the algorithms presented by Meldal et al. \cite{DBLP:conf/podc/MeldalSV91} exploit information about the communication graph, particularly information about the paths over which messages may be propagated.

\paragraph{Closely related work:}
Closest to our work is a timestamp algorithm for {\em synchronous messages} by
Garg et al. \cite{DBLP:conf/icdcs/GargS02,DBLP:journals/dc/GargSM07},
timestamps used in causal memory implementations, particularly,
{\em Lazy Replication} \cite{DBLP:journals/tocs/LadinLSG92}
and {\em SwiftCloud} \cite{zawirski2015swiftcloud}, and a hierarchical cluster timestamping
scheme \cite{Ward:2001:SHC:646666.699445}.
We will discuss these prior schemes next.

\paragraph{Synchronous messages:}
For synchronous messages, the sender process, after sending a message,
must {\em wait} until it receives an acknowledgement from the receiver process.
This constraint is exploited in \cite{DBLP:conf/icdcs/GargS02,DBLP:journals/dc/GargSM07} to design small timestamps.
In particular, if the communication network formed by the processes
is decomposed into, say, $d$ components that
are either {\em triangles} or {\em stars}, then the timestamps contain $d+4$ integer elements. Although our timestamps have similarities to the structure used
in
\cite{DBLP:conf/icdcs/GargS02,DBLP:journals/dc/GargSM07},
our algorithm does {\em not} constrain the messages to be synchronous.
As a trade-off, our timestamps are somewhat larger than \cite{DBLP:conf/icdcs/GargS02,DBLP:journals/dc/GargSM07}.
In \cite{DBLP:conf/icdcs/GargS02,DBLP:journals/dc/GargSM07}, a sender process cannot take any additional steps until it receives an acknowledgement for a sent message. We do not impose this constraint. In particular, the delay in receiving the control messages in our case only delays response to timestamp queries, but
not necessarily the computation at the processes.

\paragraph{Causal memory \cite{DBLP:journals/tocs/LadinLSG92,zawirski2015swiftcloud}:}
While there are close similarities between our work and timestamps maintained by causal memory schemes \cite{DBLP:journals/tocs/LadinLSG92,zawirski2015swiftcloud}, one critical difference is that our work focuses message passing whereas 
\cite{DBLP:journals/tocs/LadinLSG92,zawirski2015swiftcloud} focuses on shared memory. 

In {\em Lazy Replication} \cite{DBLP:journals/tocs/LadinLSG92},
each client sends its updates and queries to one of the servers. A server that receives
an update from one of the clients then propagates the update to the other servers.
Each server maintains a vector clock, similar to the timestamps at processes in our cover $\CC$: the $i$-th entry
of the vector at the $j$-server essentially counts the number of updates propagated to the $j$-th server
by the $i$-th server. In essence, the servers are {\em fully connected}, whereas our cover
$\CC$ need not be. Each client also maintains a vector similar to $vect$ in timestamps for
processes $\not\in\CC$ in our case.
Additionally, when a client sends its update to the $j$-th server,
the $j$-th element of the client's vector is updated to the index of the client's update at the $j$-th server. 
The client may potentially send the same update to multiple servers, say, $j$-th and $k$-th servers; in this case,
the $j$-th and $k$-th elements of the client's vector will be updated to the indices
of the client's updates at the respective servers.
The way the timestamps are compared in Lazy Replication differs slightly from the partial order
defined on inline timestamps, because our goal is
to capture causality {\em exactly}, whereas in Lazy Replication an approximation suffices -- this is
elaborated in Appendix \ref{a:related}.

The mechanism used in {\em SwiftCloud} \cite{zawirski2015swiftcloud} is motivated by {\em Lazy Replication} \cite{DBLP:journals/tocs/LadinLSG92},
and has close similarities to the vectors in \cite{DBLP:journals/tocs/LadinLSG92}.
In {\em SwiftCloud}, if a client sends its update to multiple servers, then the indices
returned by the servers are {\em merged} into the dependency vector maintained by the client
(optionally, some of the returned indices may not be merged).
Importantly, a server can only respond to future requests from the client provided that the server's vector
covers the client's dependency vector. This has similarities to the manner in which we compare
inline timestamps. 

The size of the timestamps in above causal memory schemes is a function of the number of servers
(that are completely connected to each other).
On the other hand, we allow arbitrary communication networks, with the size of the timestamps
depending on {\em vertex cover} size for the communication network. This enable alternative implementations.

For instance, consider a client-server architecture, wherein a large
number of clients may interact with a large number of servers. Due to the dense
communication (or interconnection) pattern in this case, the cover size
will be large, resulting in large timestamps. An alternative is illustrated in
 Figure \ref{fig:kv}, where the solid edges represent an abstract communication network.
A {\em client} or a {\em server} may communicate with multiple {\em sequencers}.
By design, the {\em sequencers} form a cover of this network. When the
number of servers is much larger than the number of sequencers, this
approach can result in a much smaller vertex cover.
In Figure \ref{fig:kv}, all communication must go through
the sequencers, and the {\em inline} timestamp size is proportional to number of sequencers. 
However, routing all server communication via sequencers can be expensive, since the sequencers
will have to handle a large volume of data.
A simple optimization can mitigate this shortcoming. For example, as shown by a dashed arrow
in Figure \ref{fig:kv}, server $R_1$ may send message contents (data) {\em directly} to server $R_2$, but
server $R_2$ will need to wait to receive {\em metadata}, in the form of timestamp information,
via sequencer $S_1$ (as shown by a dotted arrow in the figure). Thus, while the sequencers must still handle small messages to help determine
inline timestamps, bulk of the traffic can still travel between the servers directly (or similarly
between servers and clients). A similar optimization was suggested previously for
{\em totally-ordered multicast} using a sequencer \cite{Coulouris:2011:DSC:2029110}. This optimization,
in conjunction with our scheme,
provides a trade-off between timestamp size and the delay
incurred in routing metadata through sequencers.

\begin{figure}[ht]
\centering
\centerline{\includegraphics[width = 4in]{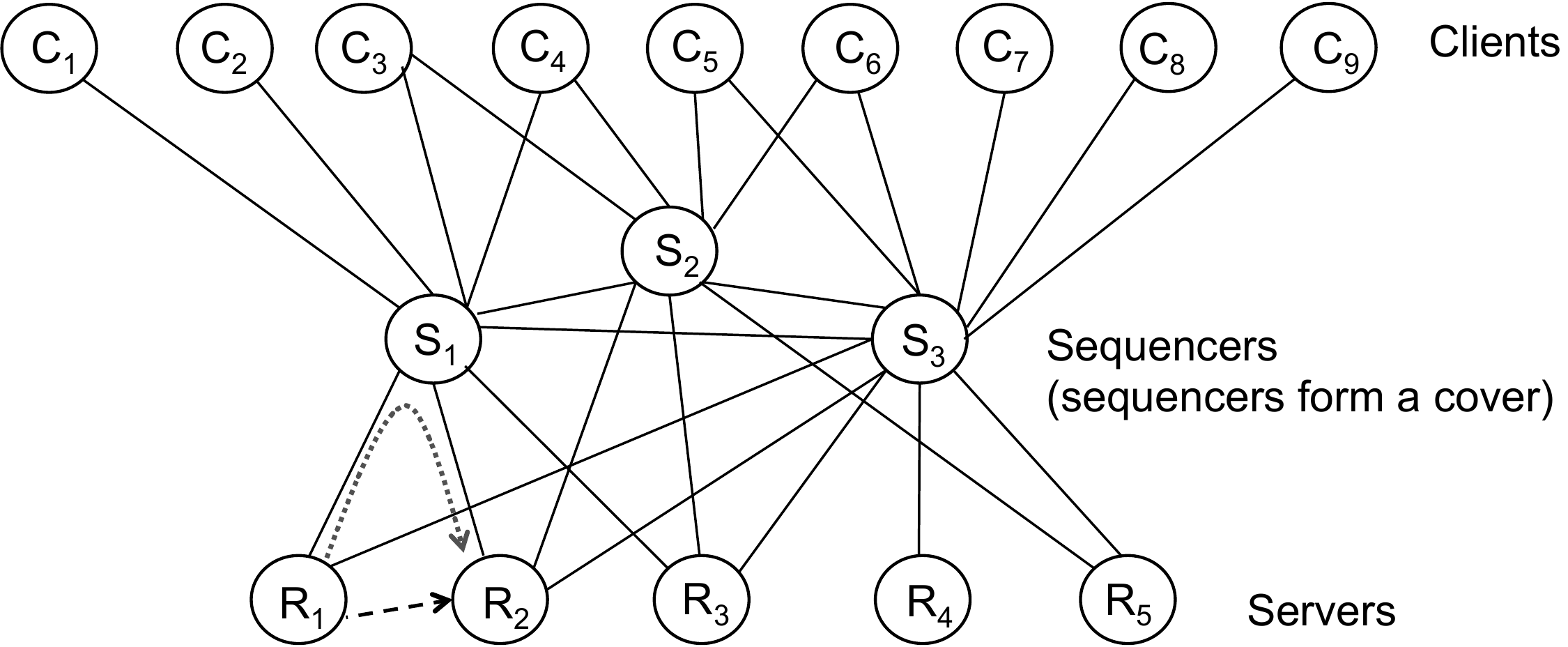}}
\caption{\em Alternate architecture suitable for inline timestamping}
\label{fig:kv}
\end{figure}

\paragraph{Cluster timestamps:}
Ward and Taylor \cite{Ward:2001:SHC:646666.699445} describe an improvement over the strategy previously proposed
by Summers, which divides the processes into clusters. They
maintain short timestamps (proportional to cluster size) for events
that occur {\em inside} the cluster, and longer timestamps (vectors with length equal to total
number of processes) for ``cluster-receive'' events.
In \cite{Ward:2001:SHC:646666.699445}, the "cluster-receive" events are assigned
long timestamps; such long timestamps are not generally necessary in our case.

\subsection{Causal Memory Systems with Smaller Timestamps}
\label{ss:sm}

As discussed previously, causal memory systems use timestamp vectors analogous to the inline timestamps discussed in this paper.
The causal memory systems maintain multiple replicas of the shared data, and
require vectors whose size is equal to the number of replicas \cite{zawirski2015swiftcloud,DBLP:journals/tocs/LadinLSG92}.
When the number of replicas is large, the vector size becomes large. To mitigate this shortcoming, we can envision a
modified architecture for causal memory, based on the idea illustrated in Figure \ref{fig:kv} for a generic client-server systems.
Each server can be viewed as a replica of the shared data. A suitable number of sequencers can be introduced to limit the size of the
timestamps. Performance may be improved by using the optimization described above. 
The prior causal memory algorithms (such as \cite{zawirski2015swiftcloud}) can be easily adapted for the architecture
in Figure \ref{fig:kv}, while incorporating timestamp objects based on inline timestamps (with size dependent on
vertex cover size, instead of the total number of servers).

\section{Summary}

We exploit the knowledge of the communication graph to
reduce timestamp size, while correctly capturing the happened-before relation.
We present an algorithm for assigning {\em inline} timestamps, and show that the
timestamps are often much smaller than vector timestamps assigned by {\em online} algorithms.
Bounds on length of vector timestamps
used by {\em online} algorithms are also presented.


\subsection*{Acknowledgements}

The authors thank Vijay Garg, Ajay Kshemkalyani and Jennifer Welch for their feedback.  We thank Marek Zawirski for answering questions about SwiftCloud \cite{zawirski2015swiftcloud}.

\bibliographystyle{abbrv}
\bibliography{\bibdir/misc,\bibdir/percolation,\bibdir/math-algo,\bibdir/nwtheory,\bibdir/manet,\bibdir/distsys,\bibdir/security,\bibdir/byzantine,\bibdir/gossip,\bibdir/sensors,\bibdir/mchannelproto,\bibdir/faultdiagnosis,\bibdir/strings,\bibdir/newwireless,\bibdir/mobile,\bibdir/vaidya,\bibdir/coding,\bibdir/chkpt_1,\bibdir/distributed,\bibdir/etc,\bibdir/ft,\bibdir/sys_1,\bibdir/voting_1,\bibdir/replicate,\bibdir/software,\bibdir/shared,\bibdir/pradhan,\bibdir/mutual,\bibdir/di,\bibdir/jennyw,\bibdir/shu,\bibdir/adhoc,\bibdir/gavin_mrate,\bibdir/esjung,\bibdir/krishna_ee,\bibdir/saad,\bibdir/smart,\bibdir/multich,\bibdir/security_shu,\bibdir/security_vartika,\bibdir/bill,\bibdir/robin,\bibdir/misbehave,\bibdir/wwt,\bibdir/clocksync,\bibdir/pradeep04march.bib,\bibdir/multichannel-cite.bib,\bibdir/temp.bib,\bibdir/rfid.bib,\bibdir/allbib-samir.bib,\bibdir/pradeepchannel.bib,\bibdir/rufus.bib,\bibdir/matt_energy.bib,\bibdir/newbib.bib,\bibdir/xue-dissertation-ref,\bibdir/lewis_paperlist,\bibdir/liang_paperlist,\bibdir/distributed_chris,\bibdir/distributed2_chris,\bibdir/consensus,\bibdir/ByzRefs,\bibdir/PSDA_DL,\bibdir/sayan1}

\newpage
\appendix

\centerline{\Large\bf Appendix}

\section{Proof of Theorem \ref{thm1}: Correctness of Inline Algorithm}
\label{sec:proof}

In this section, we prove Theorem \ref{thm1}, which claims that the timestamps provided by our inline algorithm can be used to capture causality.
Recall that partial order on inline timestamps is defined in Section \ref{sec:inline}.
For ease of reference, we define the partial order here again.

Consider two inline timestamps $\tau_1$ and $\tau_2$.
$\tau_1 < \tau_2$ if and only if one of the following is true:
\begin{list}{}{}
\item
\hspace*{0.25in} (i)~ $\tau_1.id =\tau_2.id \neq \perp$ and $\tau_1.index < \tau_2.index$, or

\hspace*{0.25in} (ii) $\tau_1.id=\tau_2.id=\perp$ and $\tau_1.vect < \tau_2.vect$, or

\hspace*{0.25in} (iii) $\tau_1.id=\perp$, $\tau_2.id\neq\perp$, and $\tau_1.vect \leq \tau_2.vect$, or

\hspace*{0.25in} (iv) $\tau_1.id\neq\perp$, $\tau_1.id\neq\tau_2.id$, and 
$\exists i,~0\leq i<c$, such that $(\tau_1.next[i] \leq \tau_2.vect[i])$.
\end{list}

\paragraph{Proof of Theorem \ref{thm1}:}

\begin{proof}

Consider events $e$ and $f$.
Let $t\geq \max(e_r,f_r)$.
Let $\tau_1 = Q^t(e)$ and $\tau_2=Q^t(f)$.

We consider four possibilities that take into account whether $e$ and $f$ occurred at processes in $\CC$ or outside $\CC$. 

~

{\em Case 1: Both $e$ and $f$ occur at processes in $\CC$: } \
In this case, $\tau_1.id = \tau_2.id = \bot$. Hence, condition (ii) above  applies. Since the processes in $\CC$ implement the standard vector clock protocol, $e \rightarrow f$ if and only if $\tau_1.vect < \tau_2.vect$

~

{\em Case 2: $e$ occurs at a process in $\CC$ and $f$ occurs at a process outside $\CC$: } \
In this case, condition (iii) applies. Also, $\tau_1.vect[i]$ (respectively, $\tau_2.vect[i]$) denotes the number of events on $p_i \in \CC$ that happened-before $e$ (respectively, $f$). Thus, $e \rightarrow f$ if and only if $f$ is aware of all events that $e$ is aware of. Note that if $e$ happened before $f$ then $f$ is aware of at least one extra event than $e$. However, this extra event may not be on a process in $\CC$. Thus, we have $e \rightarrow f$ iff $\tau_1.vect \leq \tau_2.vect$.

~

{\em Case 3: $e$ occurs at a process outside $\CC$ and $f$ occurs at a process in $\CC$: } \
In this case, $\tau_1.next[i]$ denotes the earliest time (if it exists) such that there exists an event $g_i$ on process $p_i \in 
\CC$ such that $g_i$ was created due to a message sent by the process where $e$ occurred and received by $p_i$. Since $e$ and $f$ are on different processes, $e \rightarrow f$ iff there exists $g_i$ on $p_i \in \CC$ such that $f = g_i$ or $g_i \rightarrow f$. In the former case, by construction $\tau_1.next[i]= \tau_2.vect[i]$. In the latter case, $f$ is aware of at least one extra event on $\CC$ that $g_i$ was aware of. Hence, if $e \rightarrow f$ then $\tau_1.next[i]\leq \tau_2.vect[i]$.

Also, if $\tau_1.next[i]\leq \tau_2.vect[i]$ then consider the event $g_i$ that is responsible for assignment of $\tau_1.next[i]$. Using the same argument above, $g_i = f$ or $g_i \rightarrow f$. Thus, if $\tau_1.next[i]\leq \tau_2.vect[i]$ then $e \rightarrow f$. 

~

{\em Case 4: Both $e$ and $f$ occur at processes outside $\CC$: } \
Here, we consider two cases: If $e$ and $f$ are on the same process then condition (i) applies, and, $e$ happened before $f$ iff $e$ occurred (by real time) before $f$. In other words, $e \rightarrow f$ iff $\tau_1.index < \tau_2.index$. The second subcase where $e$ and $f$ occur at different processes outside $\CC$ is similar to Case 3 except that
$f$ and $g_i$ cannot be identical in Case 4. 

\end{proof}

\section{Bounds for Vector Timestamp Length with Online Algorithms}

This appendix presents several bounds for the length of vector timestamps assigned by online
algorithms. Recall that, for vector timestamps, the $<$ partial order is defined in Section \ref{sec:vector}.

{\bf
In the discussion below, let $e_q^i$ denote the $q$-th event at process $p_i$.
}

\subsection{Lower Bound for the Star Graph}
\label{app:lower:star}

\paragraph{Real-Valued Vector Timestamps:}

Lemma \ref{l_bound} in Section \ref{sec:vector} shows that $n-1$ is a lower bound on the vector timestamp length
in star graphs when the elements of the vector may be {\em real-valued}. The lemma below derives a lower bound
when the vector elements must be integer-valued.

~

Now we consider the case when the vector elements are constrained to be integer-valued.

\paragraph{Integer-Valued Vector Timestamps:}

\begin{lemma}
\label{l_bound_integer}
Suppose that an {\em online} algorithm for the star graph assigns vector timestamps
with {\bf integer-valued} vector elements,
such that, for any two events $e,f$, $e\rightarrow f$ iff $\tau(e)<\tau(f)$.
Then the vector length must be at least $n$.
\end{lemma} 
\begin{proof}
Without loss of generality, let us assume that all vector elements of a vector timestamp must be
non-negative integers.
The proof of the lower bound is trivial for $n=1$. Now assume that $n\geq 2$.
The proof is by contradiction. Suppose that the vector length is $s\leq n-1$.

Consider an execution that includes a send event
$e_1^i$ at process $p_i$, $1\leq i\leq n-1$,
where the radial process $p_i$ sends a message to the central process $p_0$.
Let $M$ be the largest value of any of the $s$ elements of the timestamps of any of these $n-1$ send events.
Suppose that process $p_0$ initially performs $P$ computation events.
Assume that there are no other events; thus, process $p_0$ does not send any messages.
Thus, the timestamps of the send events at the radial processes cannot depend on $P$,
the number of computation events at $p_0$. Thus, we can assume that $P=(M+2)n$.
Since these $(M+2)n$ computation events occur at $p_0$ before it receives
any messages, the timestamps for these events are computed before $p_0$ learns
timestamps of any send events at the other processes. Since the timestamp elements
are constrained to be non-negative integers, one of the elements of the timestamp of the last of
these computation events at $p_0$, namely $e_{Pn}^0$ must be $>M$.
Recall that $P=(M+2)n$.

Consider set $W$ that contains event $e_{Pn}^0$ and $e_1^i,~0< i<n$.
Thus, $W$ contains $n$ events, with one event at each of the $n$ processes.
Create a set $S$ of processes as follows:
for each $l$, $0\leq l< s$,
add to $S$ any one process $p_j$ such that
the $l$-th element of the timestamp of its event in $W$ is the largest
among the $l$-th elements of the timestamps of all the events in $W$.
Clearly, $p_0\in S$ and $|S|\leq s\leq n-1$.
Consider a radial process $p_k\not\in S$ (note that $p_k\neq p_0$). Such a process
$p_k$ must exist since $|S|\leq n-1$, $p_0\in S$, and there are $n-1$ radial processes.

Suppose that the message sent by process $p_k$ at event $e_1^k$ reaches process $p_0$ after all
the other messages, including messages from the {\em radial} processes in $S$, reach process $p_0$.
Let $e=e_{Pn+n-2}^0$. By event $e$ at $p_0$, 
except for the message sent by process $p_k$, all the other messages, including messages
sent by all the radial processes in $S$, are received by process $p_0$.

Rest of the proof of this lemma is similar to the proof of Lemma \ref{l_bound}.
In particular,
define vector $E$ such that $E[l] = \max_{1\leq i<n}\, \tau(e_1^i)[l]$, $0\leq l< s$.
By definition of $W$, we also have that
$E[l] = \max_{p_i\in W}\, \tau(e_1^i)[l]$, $0\leq l< s$.
The above assumption about the order of message delivery implies that $$E \leq \tau(e_{Pn+n-2}^0).$$
Also, since $p_k\not\in W$, we have that $\tau(e_1^k)\leq E$.
This implies that $\tau(e_1^k)\leq \tau(e_{Pn+n-2}^0)$.

Since $e_1^k\neq e_{Pn+n-2}^0$, their timestamps must be distinct too. This
implies that $\tau(e_1^k)< \tau(e_{Pn+n-2}^0)$, which, in turn, implies that
$e_1^k\rightarrow e_{Pn+n-2}^0$.
However, $e_1^k$ and $e_{Pn+n-2}^0$ are concurrent events, leading to a contradiction.
\end{proof}

\subsection{Upper Bound for the Star Graph: Real-Valued Elements}
\label{app:upper:star}

For a star graph with $n=1,2$, it is easy to show that the vector length must be at least $n$,
and also that vector length $n$ suffices using the standard vector clock algorithm.
Thus, the bound of Lemma \ref{l_bound} is not tight for $ n=1, 2$.

{\bf
In the rest of this section, we focus on $n\geq 3$.
}

Lemma \ref{l_bound} shows that $n-1$ is a lower bound on the vector length used
by an online algorithm for star graphs. Now we
constructively show that this bound is tight for $n\geq 3$
by presenting an online algorithm for computing vector timestamps of length $n-1$.
The vector elements of timestamps assigned by the algorithm below
are real-valued. (If the elements are constrained to be integers, then,
as shown in Lemma \ref{l_bound_integer}, vector length of $n$ will be required.) 

We first define a function $update$ that takes process identifier $p_i$ and a vector $w$
of length $n-1$ as its arguments, and returns an updated vector.
The $n-1$ elements of the vector timestamps have indices 1 through $n-1$.
Update performed by the central process $p_0$ is different than the update performed by the radial processes.

~

\noindent
Function $update(p_k,w)$ 

\hspace*{0.25in}\{

\hspace*{0.35in} if ($p_k \neq p_0$)

\hspace*{0.6in} $w[k] :=  \lfloor w[k]+1\rfloor$ ~~~~// smallest integer larger than
		the original value of $w[k]$

\hspace*{0.35in} else

\hspace*{0.6in} for $1\leq j\leq n-1, ~~ w[j] := \mbox{any value in the
		open interval}~ (w[j], \lfloor w[j]+1\rfloor)$ 

~

\hspace*{0.25in}~~return ($w$)

\hspace*{0.25in} \}

\paragraph{Online algorithm for process $p_i$, $0\leq i<n$:}
Process $p_i$ maintains a vector $v^i$ of length $n-1$.  Initially, $v^i:={\bf 0}$.
For each event $e$ at $p_i$, perform the following steps:
\begin{enumerate}
\item If $e$ is a receive event:

\hspace*{0.25in} $u$ := vector timestamp piggybacked on message received at event $e$

\hspace*{0.25in} $v^i := \max(u,v^i)$

\item $v^i := update(p_i,v^i)$
\item $\tau(e)$ :=  $v^i$ 
\item If $e$ is a send event, then piggyback $v^i$ on the message sent at event $e$.
\end{enumerate}

Note that steps 2 and 3 above are performed for all events. Steps 1 and 4 above
are performed only for {\em receive} and {\em send} events, respectively.

\paragraph{Correctness of the Online Algorithm:} For any two events $e$ and $f$,
the algorithm assigns timestamps $\tau(e)$ and $\tau(f)$, respectively, such that $e\rightarrow f$ if and only if $\tau(e)<\tau(f)$.
The proof is straightforward and omitted here.

\subsection{Bounds for Communication Graphs with Connectivity {\Large $\kappa$}}
\label{sec:bounds:arbitrary}

%

\subsubsection{Communication graphs with vertex connectivity $\geq 2$}

\begin{lemma}
\label{l_bound_arbitrary}
Suppose that the communication graph has vertex connectivity $\geq 2$. 
For this graph, an {\em online} algorithm assigns distinct vector timestamps
to distinct events such that, for any two events $e$ and $f$, $e\rightarrow f$ if and only if $\tau(e)<\tau(f)$.
Then the vector length must be at least $n$.
\end{lemma} 
\begin{proof}
Recall that $\GG$ is the communication graph formed by the $n$ processes.

This proof is analogous to the proof of Lemma \ref{l_bound}.
The proof is trivial for $n\leq 2$.

Now assume that $n\geq 3$.
The proof is by contradiction. Suppose that the vector length is $s\leq n-1$.


Consider an execution in which,
initially, each process $p_i$, $0\leq i<n$, sends a message
to each of its neighbors in the communication graph. Subsequently, whenever a message is received from
any neighbor, a process forwards the message to all its other neighbors. 
Thus, essentially, the messages are being flooded throughout the network (the execution is infinite,
although we will only focus on a finite subset of the events).

Create a set $S$ of processes as follows: for each $l$, $0\leq l< s$,
add to $S$ any one process $p_j$ such that $\tau(e_1^j)[l] = \max_{0\leq i<n}\, \tau(e_1^i)[l]$.
Clearly, $|S|\leq s\leq n-1$.
Consider a process $p_k\not\in S$. Such a process
$p_k$ must exist since $|S|\leq n-1$.

Suppose that all the communication channels between $p_k$ and its neighbors are very slow,
but each of the remaining communication channels has a delay upper bounded by some constant $\delta>0$.
For convenience of discussion, let us ignore local computation delay between the receipt
of a message at a process and its forwarding to the neighbors. Let $D$ be defined as the maximum over the diameters
of all the subgraphs of $\GG$ containing $n-1$ vertices.
Let the delay on all communication channels of $p_k$ be $>2\delta D$.
Because the network's vertex connectivity is $\geq 2$, within duration $\delta D$,
$n-1$ processes, except $p_k$,
will have received messages initiated by those $n-1$ processes (i.e.,
all messages except the message initiated by $p_k$). 

Define vector $E$ such that $E[l] = \max_{0\leq i<n}\, \tau(e_1^i)[l]$, $0\leq l< s$.
By definition of $S$, we also have that
$E[l] = \max_{p_i\in S}\, \tau(e_1^i)[l]$, $0\leq l< s$.

Consider any process $p_i\neq p_k$.
Let $e$ be the earliest receive event at $p_i$ such that by event $e$ (i.e., including event
$e$), $p_i$ has received the messages initiated by all processes except $p_k$.
Due to the definition of $D$ and $\delta$, event $e$ occurs at $p_i$ by time $\delta D$.
Since by event $e$, $p_i$ has received the messages initiated by all other processes except $p_k$,
and $p_k\not\in S$,
we have $$E\leq \tau(e).$$


Also, since $p_k\not\in S$, we have $$\tau(e_1^k)\leq E.$$
The above two inequalities together imply that $\tau(e_1^k)\leq \tau(e)$.

Since $e_1^k$ and $e$ occur on different processes, $e_1^k\neq e$, and
their timestamps must be distinct too. Therefore,
$\tau(e_1^k)< \tau(e)$, which, in turn, implies that
$e_1^k\rightarrow e$.
However, $e_1^k$ and $e$ are concurrent events, because $e_1^k$ is the first event at $p_k$,
there are no messages received by $p_k$ before $2\delta D$, and similarly, 
no process receives messages from $p_k$ during $2\delta D$.
This results in a contradiction.
\end{proof}

~

For any graph, upper bound of $n$ is obtained by using the standard vector clock algorithm
for $n$ processes \cite{Mattern88virtualtime,DBLP:journals/computer/Fidge91}.
Thus, the bound $n$ is tight for communication graphs with vertex connectivity $\geq 2$.

~

\subsubsection{Communication graphs with vertex connectivity 1}

\begin{lemma}
\label{l_conn1}
Suppose that the communication graph has vertex connectivity $=1$. 
Define $X$ to be the set of processes such that no process in set $X$ by itself forms a vertex cut of size 1.
For this graph, an {\em online} algorithm assigns distinct vector timestamps
to distinct events such that, for any two events $e$ and $f$, $e\rightarrow f$ if and only if $\tau(e)<\tau(f)$.
Then the vector length must be at least $|X|$.
\end{lemma} 
\begin{proof}

Recall that $\GG$ is the communication graph formed by the $n$ processes.

This proof is analogous to the proof of Lemma \ref{l_bound_arbitrary}.
The proof is trivial for $|X|=1$.

Now assume that $|X|\geq 2$. The proof is by contradiction. Suppose that the vector length is $s\leq |X|-1$.


Consider an execution in which,
initially, each process $p_i\in X$ sends a message
to each of its neighbors in the communication graph. Subsequently, whenever a message is received from
any neighbor, a process forwards the message to all its other neighbors. 
Thus, essentially, the messages initiated by processes in $X$ are being flooded throughout the network (the execution is infinite,
although we will only focus on a finite subset of the events).

Create a set $S$ of processes as follows: for each $l$, $0\leq l< s$,
add to $S$ any one process $p_j\in X$ such that $\tau(e_1^j)[l] = \max_{p_i\in X}\, \tau(e_1^i)[l]$.
Clearly, $|S|\leq s\leq |X|-1$.
Consider a process $p_k\in X$ such that $p_k\not\in S$. Such a process
$p_k$ must exist since $|S|\leq |X|-1$.

Suppose that all the communication channels between $p_k$ and its neighbors are very slow,
but each of the remaining communication channels has a delay upper bounded by some constant $\delta>0$.
For convenience of discussion, let us ignore local computation delay between the receipt
of a message at a process and its forwarding to the neighbors. Let $D$ be defined as the maximum over the diameters
of all the subgraphs of $\GG$ containing all vertices except any one vertex in $X$ (there are $|X|$ such subgraphs). By definition of $X$,
removing any one process in $X$ from the graph $\GG$ will not partition the subgraph.
Let the delay on all communication channels of $p_k$ be $>2\delta D$.
Within duration $\delta D$, all $n-1$ processes, except $p_k$,
will have received the messages initiated by the $|X|-1$ processes in $X-\{p_k\}$.

Define vector $E$ such that $E[l] = \max_{p_i\in X}\, \tau(e_1^i)[l]$, $0\leq l< s$.
By definition of $S$, we also have that
$E[l] = \max_{p_i\in S}\, \tau(e_1^i)[l]$, $0\leq l< s$.

Consider any process $p_i\in X$ such that $p_i\neq p_k$.
Let $e$ be the earliest receive event at $p_i$ such that by event $e$ (i.e., including event
$e$), $p_i$ has received the messages initiated by all processes except $p_k$.
Due to the definition of $D$ and $\delta$, event $e$ occurs at $p_i$ by time $\delta D$.
Since by event $e$, $p_i$ has received the messages initiated by all other processes except $p_k$,
and $p_k\not\in S$,
we have $$E\leq \tau(e).$$


Also, since $p_k\not\in S$, we have $$\tau(e_1^k)\leq E.$$
The above two inequalities together imply that $\tau(e_1^k)\leq \tau(e)$.

Since $e_1^k$ and $e$ occur on different processes, $e_1^k\neq e$, and
their timestamps must be distinct too. Therefore,
$\tau(e_1^k)< \tau(e)$, which, in turn, implies that
$e_1^k\rightarrow e$.
However, $e_1^k$ and $e$ are concurrent events, because $e_1^k$ is the first event at $p_k$,
there are no messages received by $p_k$ before $2\delta D$, and similarly, 
no process receives messages from $p_k$ during $2\delta D$.
This results in a contradiction.
\end{proof}

~

Observe that for star graph, vertex connectivity is 1, and $X$ consists of all the radial
processes. Thus, $|X|=n-1$.

~

For a communication graph with vertex connectivity 1,
an upper bound of $n-1$ (not necessarily tight) is obtained by assigning the
role of $p_0$ in the star graph to any one
process that forms a cut of the communication graph, and then using the
vector timestamping algorithm presented previously for the star graph.
In general, there is a gap between the above upper bound of $n-1$ and lower bound of $|X|$.
It is presently unknown whether $|X|$ is a tight bound for online algorithms that assign
vector timestamps.

\section{Vector Length 2 Insufficient for Star Graph with Offline\\ Algorithms}
\label{a:offline}

Results presented above show that, for certain graphs,
including a star graph, vector timestamps of length $n-1$ or $n$ are
required when using {\em online} algorithm.
In particular, for a star graph, with real-valued vectors, vector timestamp length of $n-1$
is required.
Recall that, for vector timestamps, we use the partial order $<$ defined in Section \ref{sec:vector}.

This section considers whether the requirement can be reduced with
an offline timestamp algorithm for {\em star} graphs. As
noted in Section \ref{sec:related}, it is known that
for complete networks vectors of length $n$ are required in general.
However, it is not clear whether smaller length may suffice for offline algorithms
for restricted graphs, such as the star graph. 
Here we take a small step in resolving this question.
In particular, we consider a star graph with 4 processes, and show that a vector of length at least
3 is required even when using an {\em offline} algorithm.
Extension of this result to a star graph with larger number of processes is presently an open problem.


\begin{theorem}
Given a system of 4 processes, there does not exist an offline algorithm that assigns each event $e$ a vector $vc_e$ of size 2 such that 
\begin{tabbing}
\hspace*{5mm} \= $e \rightarrow f$\\
iff\\
\> $(vc_e[0] \leq vc_f[0] \wedge vc_e[1] \leq vc_f[1]) \wedge$\\
\> $(vc_e[0] < vc_f[0] \vee vc_e[1] < vc_f[1]) $
\end{tabbing}
\end{theorem}
\begin{proof}
To prove this theorem, we generated a counterexample with guidance from SMT solver Z3 \cite{z3}. Specifically, given a communication diagram, for any two events, we introduce constraints based on whether the pair satisfies the happened-before relation or not. Subsequently, we use Z3 to check if those constraints are satisfied. For the communication diagram in Figure \ref{fig:novcsize2}, Z3 declares that satisfying all the constraints is impossible. (The set of constraints for this diagram are available at {\tt http://www.cse.msu.edu/\~ \ }$\!\!${\tt sandeep/NoVCsize2/}) In other words, it is impossible to assign timestamps of size 2 for the communication diagram in Figure \ref{fig:novcsize2}. Thus, the above theorem follows.
\end{proof}

Subsequent to finding the example communication diagram using Z3, we also carried out a manual proof
that vector length of 2 is insufficient.

\begin{figure}[tbhp]
\centering
{\includegraphics[height = 2in]{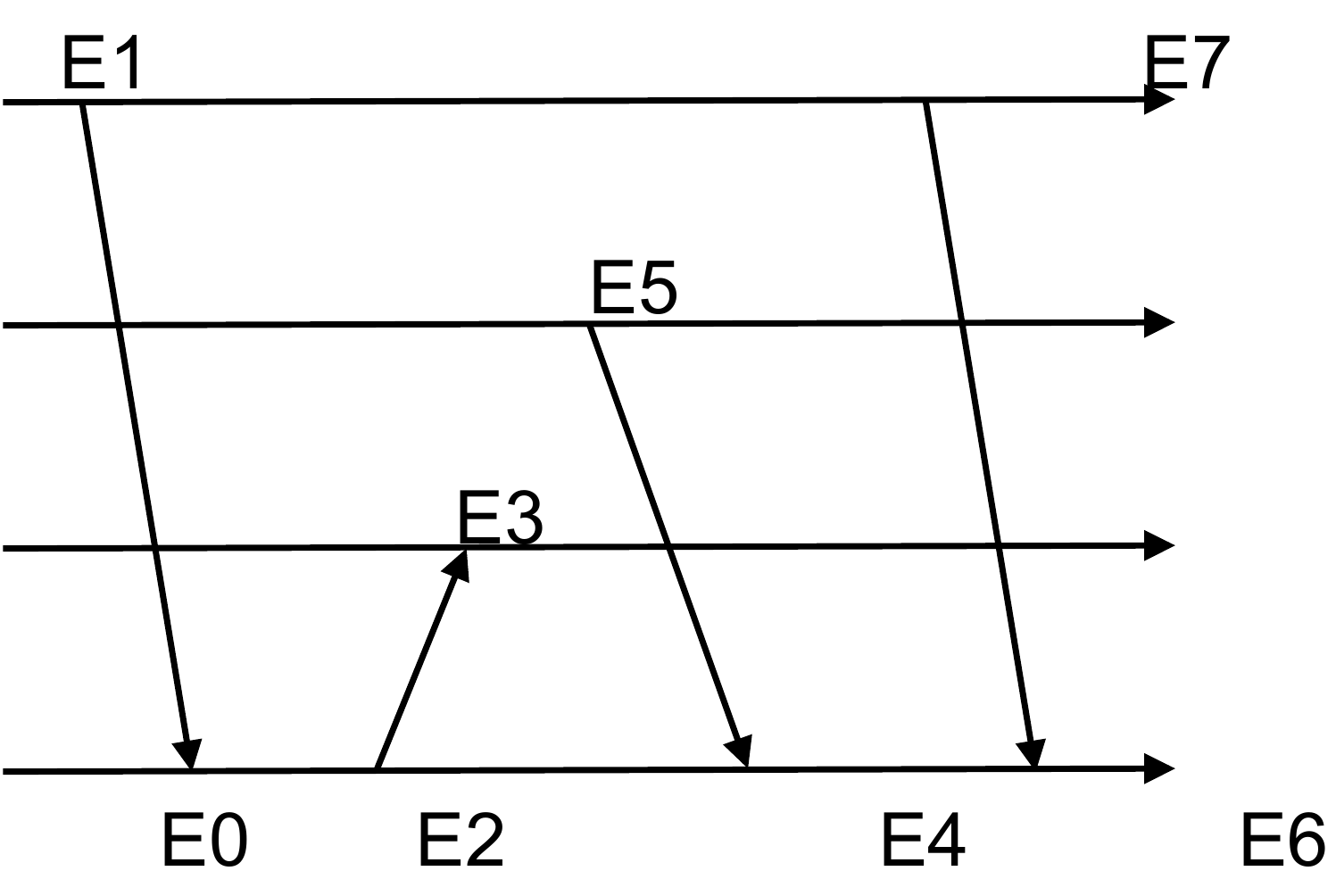}}
\caption{Vector length of 2 is insufficient}
\label{fig:novcsize2}
\end{figure}

\section{Related Work in \cite{DBLP:conf/icdcs/GargS02,DBLP:journals/dc/GargSM07,DBLP:journals/tocs/LadinLSG92,zawirski2015swiftcloud}}
\label{a:related}

Related work is discussed in Section \ref{sec:related}. In this section, we expand
on the discussion of the work in \cite{DBLP:conf/icdcs/GargS02,DBLP:journals/dc/GargSM07,DBLP:journals/tocs/LadinLSG92,zawirski2015swiftcloud}, which is most relevant to this paper. In particular, our inline timestamps have close
similarities to comparable objects in these prior papers.

\paragraph{Synchronous messages \cite{DBLP:conf/icdcs/GargS02,DBLP:journals/dc/GargSM07}:}
For synchronous messages, the sender process, after sending a message,
must {\em wait} until it receives an acknowledgement from
the receiver process, as illustrated in Figure \ref{fig:synchronous}.
This constraint is exploited in \cite{DBLP:conf/icdcs/GargS02,DBLP:journals/dc/GargSM07} to design small timestamps.
In particular, if the communication network formed by the processes
is decomposed into, say, $d$ components that
are either {\em triangles} or {\em stars}, then the timestamps contain $d+4$ integer elements.
Due to the synchronous nature of communication, messages within each component are totally ordered. The timestamps in \cite{DBLP:conf/icdcs/GargS02,DBLP:journals/dc/GargSM07} exploit this total ordering, such that the $j$-th element of a vector included in the timestamp for an event represents the number of messages within the $j$-th component (of the
decomposition) that happened before the given event.  

\begin{figure}[ht]
\centering
\centerline{\includegraphics[height = 1.6in]{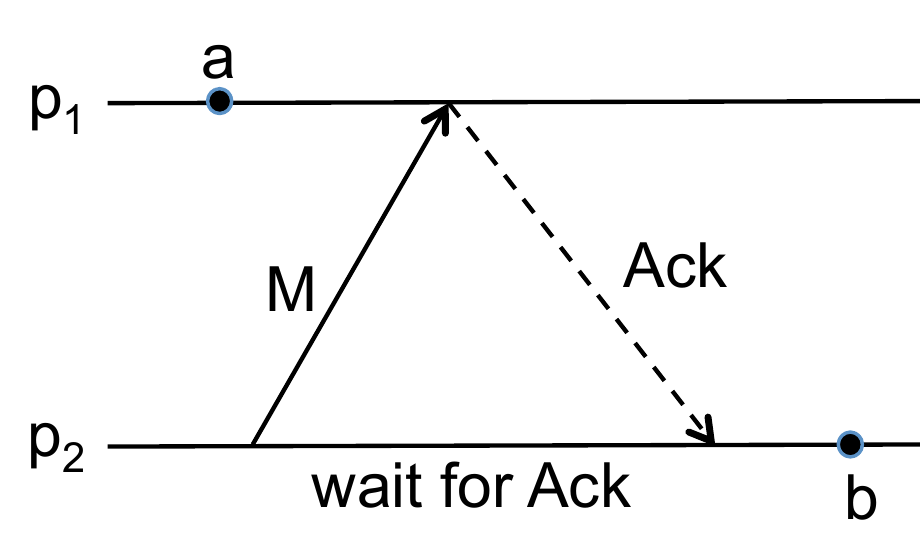}}
\caption{Synchronous message}
\label{fig:synchronous}
\end{figure}

Our timestamping algorithm does {\em not} constrain the messages to be synchronous.
Our approach has some similarities to \cite{DBLP:conf/icdcs/GargS02,DBLP:journals/dc/GargSM07} and also some
key differences.
In our case, the timestamp contains $2c+2$ integer elements, where $c$ is the size of the
vertex cover of the communication network formed by the processes.
Thus, the timestamps contain more elements
because we allow the flexibility of using asynchronous messages.
A consequence of allowing asynchronous messages is that the $next$ field of our inline timestamp may need
to be modified up to $c$ times as the execution progresses, where $c=|\CC|$
is the size of the chosen vertex cover of the communication graph.
When vertex cover size is $c$, the network can be decomposed into $c$ stars. However, our algorithm does not utilize the decomposition as such (but instead uses the knowledge of a cover set $\CC$). On the other hand, the algorithm in \cite{DBLP:conf/icdcs/GargS02,DBLP:journals/dc/GargSM07} explicitly uses the decomposition into triangle and stars. 

The $next$ field in our timestamp for events outside $\CC$ includes an index for
the receive event of one message sent to each of the $c$ processes in the vertex cover.
Thus, the $next$ field includes $c$ elements.
On the other hand,
the timestamps in \cite{DBLP:conf/icdcs/GargS02,DBLP:journals/dc/GargSM07} include
just 1 index that has functionality analogous to one of the $c$ elements in our
$next$ field. This index in the timestamp in \cite{DBLP:conf/icdcs/GargS02,DBLP:journals/dc/GargSM07} counts messages in a component of the edge decomposition,
whereas in our case, the index counts number of events at a process.
These distinctions are caused by the restriction of synchronous messages
in \cite{DBLP:conf/icdcs/GargS02,DBLP:journals/dc/GargSM07}, and allowance for
asynchronous messages in our scheme.

\paragraph{Causal memory \cite{DBLP:journals/tocs/LadinLSG92,zawirski2015swiftcloud}:}
The purpose of the timestamps used in the work on causal memory is to ensure causal consistency.
There are close similarities between our timestamps and comparable objects maintained in some causal memory
schemes \cite{DBLP:journals/tocs/LadinLSG92,zawirski2015swiftcloud}.

In {\em Lazy Replication} \cite{DBLP:journals/tocs/LadinLSG92},
a client-server architecture is used to implement causally consistent shared memory.
Each server maintains a copy of the shared memory.
Each client sends its updates and queries to one of the servers. A server that receives
an update from one of the clients then propagates the update to the other servers.
Each server maintains a vector clock: the $i$-th entry
of the vector at the $j$-server essentially counts the number of updates propagated to the $j$-th server
by the $i$-th server. Each client also maintains a similar vector:
the $i$-th element of the client's vector counts the number of updates propagated by the $i$-th server
on which the client's state depends. Additionally, when a client sends its update to the $j$-th server,
the $j$-th element of the client's vector is updated to the index of the client's update at the $j$-th server. 
The client may potentially send the same update to multiple servers, say, $j$-th and $k$-th servers; in this case,
the $j$-th and $k$-th elements of the client's vector will be updated to the indices
of the client's updates at the respective servers.
Finally, a server cannot process an update or a query from a client until the server's vector 
clock is $\geq$ the vector at the client. The $\geq$ operator here performs an element wise
comparison of the vector elements: vector $v\geq w$ only if each element of $v$ is $\geq$ the
corresponding element of vector $w$. 
This comparison operation differs slightly from the way we compare analogous elements in our timestamps
in partial order $>$ for inline timestamps, as defined towards the end of Section \ref{sec:inline}.
In particular, recall condition (iv) of the partial order $<$ from Section \ref{sec:inline}.
In condition (iv), it suffices to satisfy inequality for {\em just one} element of the  $next$ array. This is in contrast
to vector comparison used in Lazy Replication.

The mechanism used in {\em SwiftCloud} \cite{zawirski2015swiftcloud} is motivated by {\em Lazy Replication} \cite{DBLP:journals/tocs/LadinLSG92},
and has close similarities to the vectors in \cite{DBLP:journals/tocs/LadinLSG92}.
In {\em SwiftCloud}, if a client sends its update to multiple servers, then the indices
returned by the servers are {\em merged} into the dependency vector maintained by the client
(optionally, some of the returned indices may not be merged).
Importantly, a server can only respond to future requests from the client provided that the server's vector
covers the client's dependency vector.


Beyond some small differences in how the timestamp comparison is performed, the other difference
between shared memory schemes above and our solution is that the above schemes rely on a set
of {\em servers} through which the processes interact with each other. Thus, the
communication network in their case is equivalent to a clique of servers to which the clients
are connected.  The size of the timestamps
is a function of the number of servers. On the other hand, we allow arbitrary communication networks, with the size of the timestamps being a function of a {\em vertex cover} for the communication
network. The vertex cover is not necessarily completely connected.
Secondly, dependencies introduced through {\em events} happening at the servers (e.g.,
receipt of an update from a client) in the shared memory systems are not
necessarily {\em true} dependencies. For instance, suppose that process $p_0$ propagates
update to variable $x$ to replica R, then process $p_1$ propagates update
to variable $y$ to replica R, and finally process $p_2$ reads updated value of $y$ from
replica R. In the shared memory dependency tracking schemes above, the update of $x$ by $p_0$
would be treated as having happened-before the read by $p_2$. In reality, there is no such
causal dependency. But the dependency is introduced artificially as a cost of reducing the
timestamp size. On the other hand, in the message-passing context, if the communication
network reflects the communication channels used by the processes, then no such artificial
dependencies will arise. However, in the message-passing case as well, we can introduce
artificial dependencies by disallowing the use of certain communication channels in order to
decrease the vertex cover size. This was illustrated in Section \ref{sec:related} through the example
in Figure \ref{fig:kv}.

\section{Implementation Issues}
\label{sec:app}

In the inline algorithm, recall that elements of the $next$ field of the timestamps
of events at processes outside $\CC$ may have to be modified as many as $c$ times, where
$c$ is the size of the vertex cover chosen by the algorithm.

In particular, when a process $p_i\not\in\CC$ sends a message to some process $p_j\in\CC$
at some time $t$, the $j$-th element of $next$ field in $\tau(e)$ as well as 
$next[j]$ element for any prior event $f$ for which $\tau(f).next[j]=\infty$, is modified
to equal to the index of the receive event at $p_j$ corresponding to the message sent at $e$.
Before the modification can be made, there is a delay due to the wait for a control message
from $p_j$ that will inform $p_i$ of this index.
Any queries at time $\geq t$ for the timestamps of event $e$, and events such as $f$, should
not return until the index is known. To implement this, when the message is sent at $e$,
the $next[j]$ element of $e$ and $f$ can be set to $\perp$ to indicate an invalid value -- when
such an invalid value is found in $next$ field for an event, the query procedure will know
that it must wait for the invalid value to be updated before the event's timestamp can be returned.

Secondly, consider the set of events, $E_m^{i,j}$, at a process $p_i\not\in\CC$ that occur between
the $m$-th and $m+1$-th messages sent by $p_i$ to a process $p_j\in\CC$.
Observe that for all the events in $E_m^{i,j}$ the $next[j]$ element of their timestamps
is identical. This fact can be exploited by process $p_i$ to make it easier to update the
timestamps stored at $p_i$. In particular, for all the events in $E_m^{i,j}$, the $next[i]$
element of the timestamp can point to an identical memory location -- modifying this location then
modifies the $j$-th element of all these timestamps simultaneously.

Two other improvements can be made to the $next$ component of the timestamp:
\begin{itemize}
\item {\em Reducing the size of the $next$ field}: In our discussion so far
we assume
that the $next$ field of the timestamp of an event outside $\CC$ includes one element per process in $\CC$.
However, it suffices for the $next$ field for timestamps at process $p_i\not\in\CC$
to include an element for each neighbor of $p_i$ in $\CC$.
Since process $p_i\not\in\CC$ never sends a message directly to any process $p_j\in\CC$
such that $(p_i,p_j)\not\in\scripte$, the elements of $next$ corresponding to such $p_j$
in timestamps for events at $p_i$ will always remain $\infty$. Hence these elements can be
safely removed from the $next$ field.
Thus reduces the size of the $next$ field for events at process $p_i\not\in\CC$ to
the number of its neighbor processes (which are necessarily all in $\CC$).

\item {\em Reducing the delay in computing the $next$ field elements:}
In the basic algorithm presented in Section \ref{sec:inline},
for an event $e$ at process $p_i\in\CC$, the $j$-th element of the $next$
field cannot be computed (where $(p_i,p_j)\in\scripte$) until a message
from $p_i$ (sent at $e$ or later) is received by $p_j$.
The essential use of the $next[j]$ field is to learn the index of the
earliest event at $p_j$ that is ``directly'' influenced by event $e$ (``directly''
here means due to a message from $p_i$ to $p_j$).

Now we suggest a potential alternative, illustrated in Figure \ref{fig:improve}.
In Figure \ref{fig:improve}(a), both the elements of the $next$ field of event $e$ at $p_3$
are $\infty$. As shown in Figure \ref{fig:improve}(c), although $p_3$ does not send a message
to $p_1$ at or after event $f$, event $f$ does influence event $g$ at process $p_1$.
In this case, it would be acceptable if $next[1]$ element of 
$\tau(e)$ and $\tau(f)$ is set equal to 2 (because $g$ is the second event at $p_1$).
However, for $p_3$ to be able to learn this event index, additional control information will have to be
exchanged between the processes. The benefit of the optimization is that the $next$ elements are
changed from $\infty$ to finite values earlier than the basic approach illustrated earlier,
but potentially at the cost of greater control overhead.
A detailed design of this solution is not yet developed.


\begin{figure}[p]
\centering
\centerline{\includegraphics[height = 2.4in]{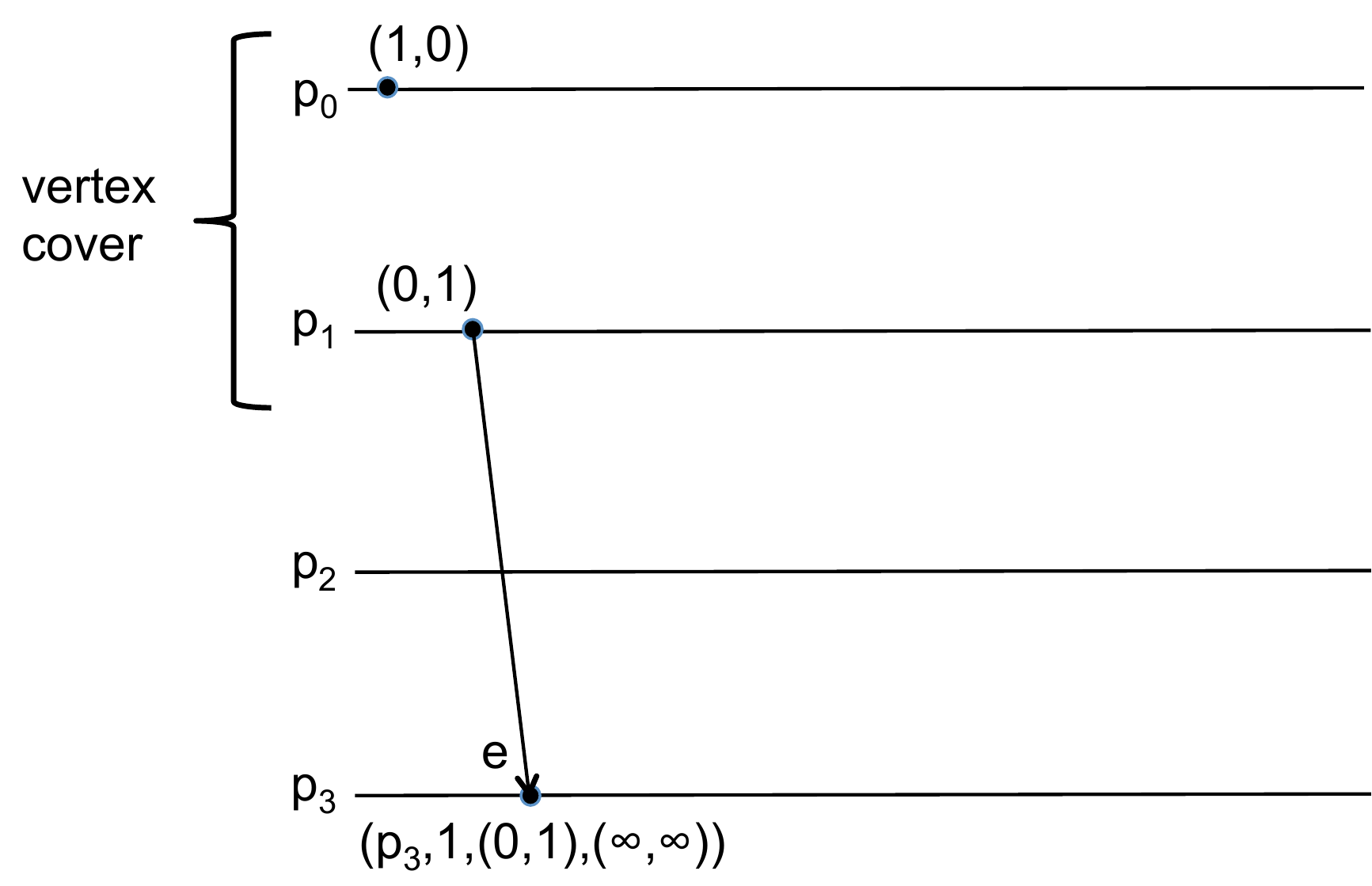}}

\centerline{(a)}

~

~

\centerline{\includegraphics[height = 2.4in]{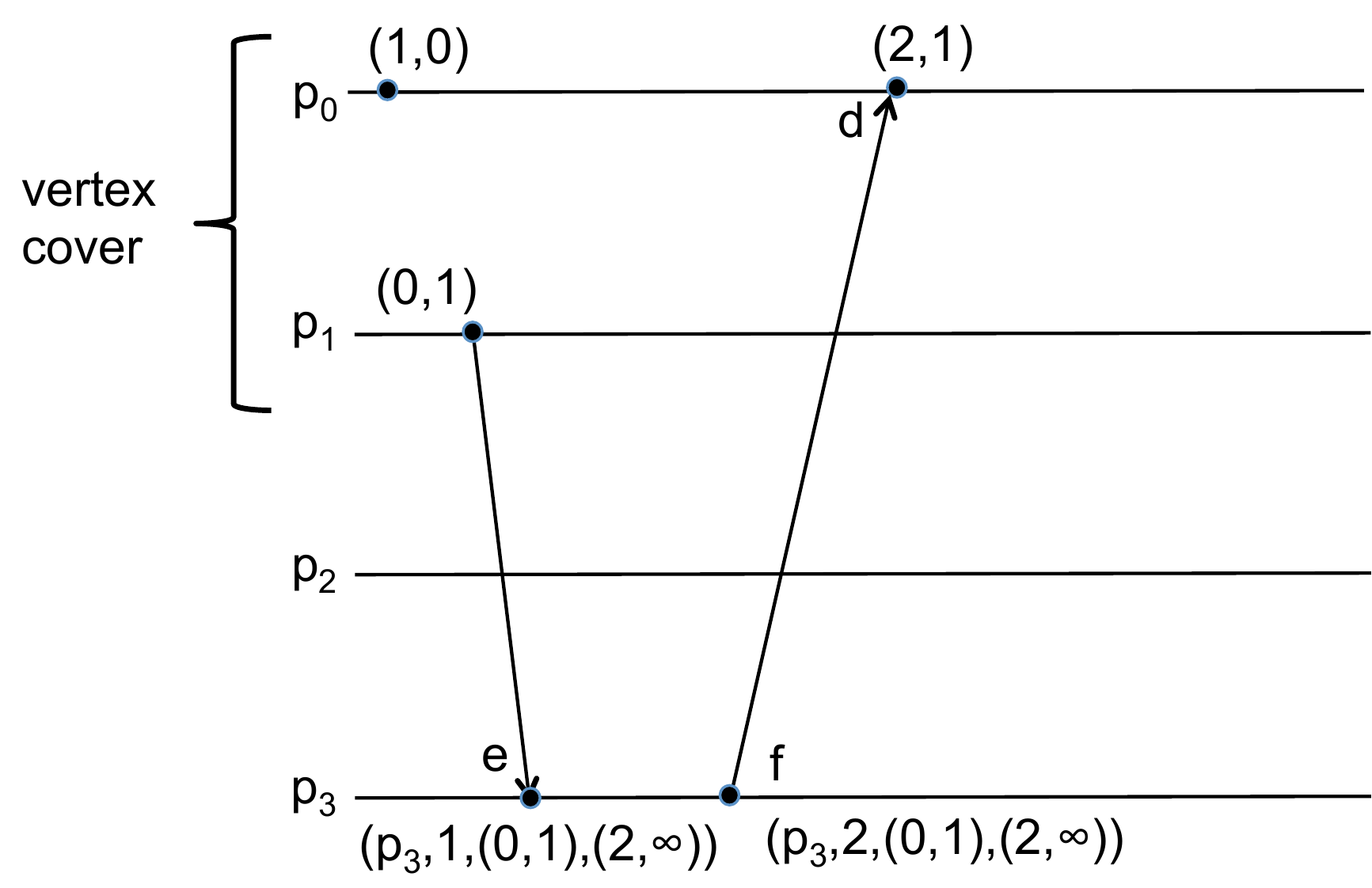}}

\centerline{(b)}

~

~

\centerline{\includegraphics[height = 2.4in]{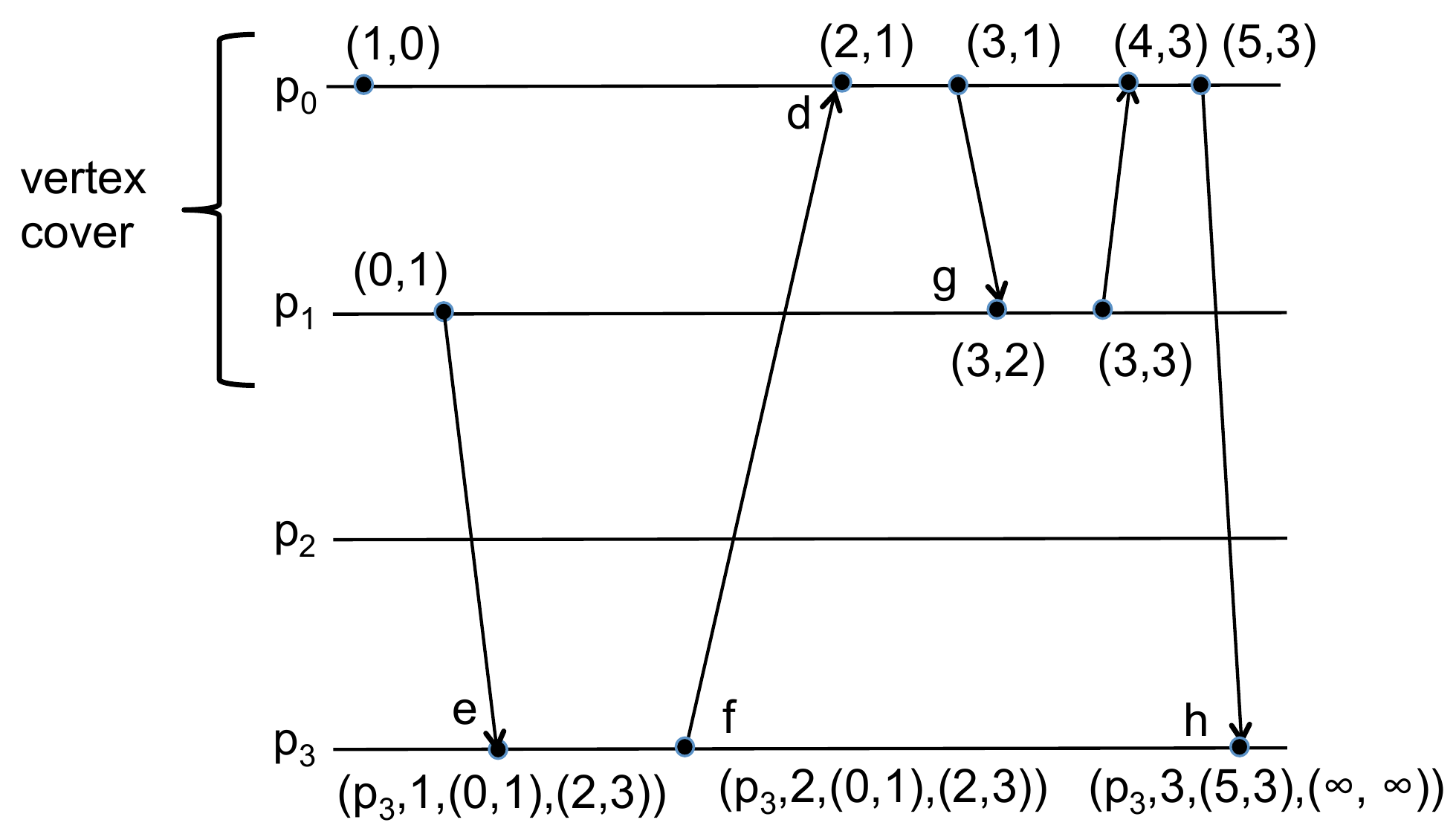}}

\centerline{(c)}

\caption{Improvement in $next$ computation}
\label{fig:improve}
\end{figure}

\end{itemize}

\end{document}